\documentclass[a4paper,twoside,11pt]{article}
\usepackage[margin=3cm]{geometry}
\usepackage[affil-it]{authblk}
\usepackage{silence}
\ActivateWarningFilters[pdftoc]

\usepackage[english]{babel}
\usepackage{multirow}
\usepackage{hyphenat}
\usepackage{amsmath,amssymb,amsfonts,amsthm}
\usepackage{booktabs,tabularx,enumerate}
\usepackage{url,xspace,soul,wrapfig,lineno}
\usepackage{graphicx,paralist}
\usepackage[font=small, hypcap=true]{subfig,caption}
\usepackage[textsize=footnotesize]{todonotes}
\usepackage[inline]{enumitem}

\newcommand{\fb}{$1$-fbp\xspace}
\newcommand{\fblong}{$1$-fan-bundle-planar\xspace}
\newcommand{\fby}{$1$-fan-bundle-planarity\xspace}
\newcommand{\myparagraph}[1]{\medskip\noindent\textbf{#1}}
\newcommand{\threepartition}{{\sc 3-Par\-ti\-tion}\xspace}

\newtheorem{theorem}{Theorem}
\newtheorem{lemma}[theorem]{Lemma}

\newcounter{casecounter}
\newcounter{subcasecounter}
\newcounter{subsubcasecounter}

\makeatletter
\newcommand{\ccase}[1]{%
  \stepcounter{casecounter}%
  \setcounter{subcasecounter}{0}%
  \setcounter{subsubcasecounter}{0}%
  \bgroup
  \setlength{\parindent}{0em}
  \protected@write \@auxout {}{\string \newlabel {#1}{{\thecasecounter}{\thepage}{\thecasecounter}{#1}{}} }%
\medskip
\textbf{Case \thecasecounter:}%
  \egroup
}

\newcommand{\subcase}[1]{%
  \stepcounter{subcasecounter}%
  \setcounter{subsubcasecounter}{0}%
  \bgroup
  \protected@write \@auxout {}{\string \newlabel {#1}{{\thecasecounter.\thesubcasecounter}{\thepage}{\thecasecounter.\thesubcasecounter}{#1}{}} }%
\smallskip
\textbf{Case \thecasecounter.\thesubcasecounter:}%
  \egroup
}

\newcommand{\subsubcase}[1]{%
  \stepcounter{subsubcasecounter}%
  \bgroup
  \protected@write \@auxout {}{\string \newlabel {#1}{{\thecasecounter.\thesubcasecounter.\thesubsubcasecounter}{\thepage}{\thecasecounter.\thesubcasecounter.\thesubsubcasecounter}{#1}{}} }%
\textbf{Case \thecasecounter.\thesubcasecounter.\thesubsubcasecounter:}%
  \egroup
}

\usepackage{hyperref}

\graphicspath{{figures/}}

\begin{document}
\title{1-Fan-Bundle-Planar Drawings of Graphs}
\author[1]{P.~Angelini}
\author[1]{M.~A. Bekos}
\author[1]{M.~Kaufmann}
\author[2]{P.~Kindermann}
\author[1]{T.~Schneck}
\affil[1]{Institut f{\"u}r Informatik, Universit{\"a}t T{\"u}bingen, Germany\\
$\{$angelini,bekos,mk,schneck$\}$@informatik.uni-tuebingen.de}
\affil[2]{LG Theoretische Informatik, FernUniversit\"at in Hagen, Germany\\
philipp.kindermann@fernuni-hagen.de}
\date{}

\maketitle

\begin{abstract}
Edge bundling is an important concept, heavily used for graph visualization purposes. To enable the comparison with other established nearly-planarity models in graph drawing, we formulate a new edge-bundling model which is inspired by the recently introduced fan-planar graphs. In particular, we restrict the bundling to the endsegments of the edges. Similarly to 1-planarity, we call our model \emph{1-fan-bundle-planarity}, as we allow at most one crossing per bundle.

For the two variants where we allow either one or, more naturally, both endsegments of each edge to be part of bundles, we present edge density results and consider various recognition questions, not only for general graphs, but also for the outer and 2-layer variants. We conclude with a series of challenging questions.
\end{abstract}

\section{Introduction}
\label{sec:introduction}

Edge bundling is a powerful tool used in information visualization to avoid
visual clutter. In fact, when the edge density of the network is too high,
the traditional techniques of graph layouts and flow maps become unusable. In
this case, grouping together parts of edges that flow parallel to each other
within a single bundle allows us to reduce the clutter and improve readability; see Fig.~\ref{fig:edgebundling} for an example.
Among the many, we mention here only the seminal papers of
Holten~\cite{DBLP:journals/tvcg/Holten06} and Telea and
Ersoy~\cite{DBLP:journals/cgf/TeleaE10}, which focus on radial layouts, as well
as works on flow maps~\cite{DBLP:journals/tvcg/BuchinSV11} and parallel
coordinates~\cite{DBLP:journals/cgf/ZhouYQCC08}. 
For a comprehensive overview and
an evaluation refer to Zhou et al.~\cite{zhou13}.
Confluent drawings~\cite{degm-cdvnd-GD03} represent edges by planar curves that are not
interior-disjoint, so the parts that are used by several edges can be interpreted as~bundles,
and in contrast to other edge
bundling techniques, they are not ambiguous.

\begin{figure}[t]
    \centering
    \includegraphics[width=0.9\textwidth]{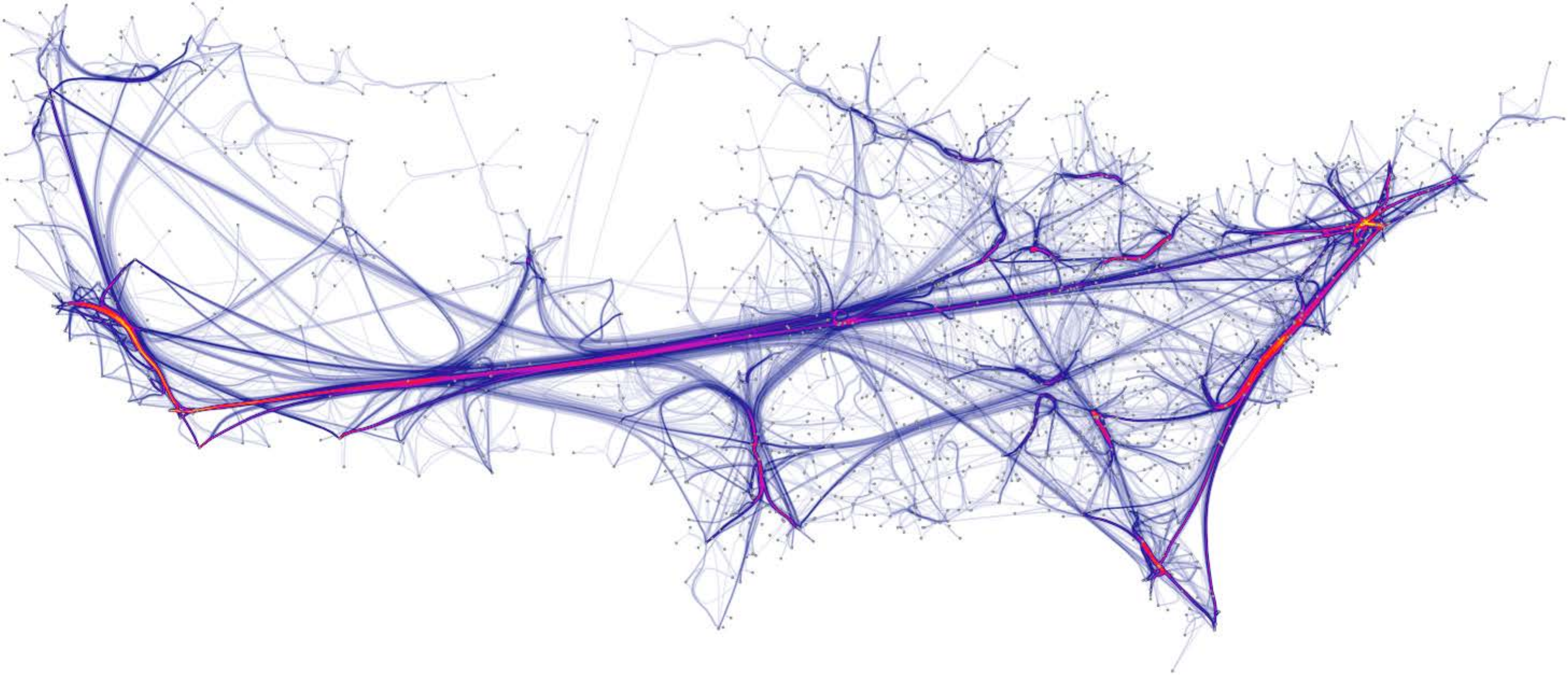}
    \caption{An example of edge bundling (taken from~\cite{DBLP:journals/cgf/HoltenW09}).} 
    \label{fig:edgebundling}
\end{figure}

In this work, we combine the powerful
visualization technique of edge bundling with previous theoretical considerations from
the area of \emph{beyond-planarity}, which is currently receiving
strong attention (see,
e.g.,~\cite{SoCG2017,Shonan2016,Dagstuhl2016}). This area focuses on
drawings of graphs in which in addition to a planar graph structure
some crossings may be allowed, if they are limited to locally defined
configurations. Different constraints on the crossing configurations
define different \emph{nearly-planar} graph classes. Classical
examples are \emph{$1$-planar} graphs~\cite{MR0187232}, which
allow for drawings in which each edge is crossed at most once, and
\emph{quasi-planar}
graphs~\cite{DBLP:journals/combinatorica/AgarwalAPPS97}, which admit
 drawings not containing three mutually crossing~edges.

\begin{figure}[b]
    \centering
    \subfloat[\label{fig:fanplanar}{}]
    {\includegraphics[page=1]{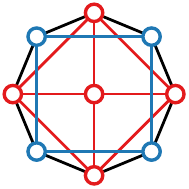}}
    \hfil
    \subfloat[\label{fig:bundleplanar}{}]
    {\includegraphics[page=2]{introduction}}
    \hfil
    \subfloat[\label{fig:k6}{}]
    {\includegraphics[page=3]{introduction}}
    \hfil
    \subfloat[\label{fig:k5}{}]
   {\includegraphics[page=4]{introduction}}
    \caption{(a--b)~The fan-planar graph of (a) is redrawn in (b) under the $1$-sided model,
    (c)~a $2$-sided \fb drawing of $K_6$,
    (d)~a $1$-sided \fb drawing of $K_5 \setminus e$ (the missing edge is drawn dotted).}
    \label{fig:introduction}
\end{figure}

Another typical example of nearly-planar graphs is the class of \emph{fan-planar}
graphs~\cite{DBLP:journals/corr/KaufmannU14}. In a
\emph{fan-planar drawing}~\cite{BekosCGHK14,DBLP:conf/gd/BinucciCDGKKMT15,
DBLP:journals/tcs/BinucciGDMPST15,DBLP:journals/corr/KaufmannU14}, an edge is allowed to cross
multiple edges as long as they belong to the same \emph{fan}, that is, they are
all incident to the same vertex; refer to Fig.~\ref{fig:fanplanar}. Such a
crossing is called a \emph{fan crossing}. The idea is that edges incident to
the same vertex are somehow close to each other, and thus having an edge
crossing all of them does not affect readability too much. In other words,
edges of a fan can be grouped into a \emph{bundle} such that the crossings
between an edge and all the edges of the fan become a single crossing.
In Fig.~\ref{fig:bundleplanar} we
show the bundle-like edge routing corresponding to the fan-planar drawing in
Fig.~\ref{fig:fanplanar}. Note, however, that the original definition of
fan-planar drawings does not always allow for this type of bundling, as in the
case of graph $K_{4,n-4}$, for large enough $n$ (see Section~\ref{sec:relationships}).

We thus introduce \emph{\fblong} drawings (\emph{\fb} drawings for short), in
which edges of a fan can be bundled together and crossings between bundles
are allowed as long as each bundle is crossed by at most one other bundle;
see Figs.~\ref{fig:bundleplanar}--\ref{fig:k5}. More formally, in a \fb
drawing every edge has $3$ parts: the first and the last parts are
\emph{fan-bundles}, which may be shared by several edges, while the middle part is
\emph{unbundled}. Each fan-bundle can cross at most one other fan-bundle, while the
unbundled parts are crossing-free. We remark that fan-bundles are not allowed
to branch, that is, each fan-bundle has exactly two end points: one of
them is the vertex the fan is incident to, while at the other one all the
edges in the fan are separated from each other.

The latter ``1-planarity'' restriction prevents a fan-bundle of an edge to
cross edges of several fans, which would be not allowed in a fan-planar
drawing. However, since every edge has two fan-bundles, each of which can
cross another bundle, it is still possible that an edge crosses two different
fans, hence making the drawing not fan-planar. In order to avoid this, we
introduce a restricted model of \fb drawings, called \emph{$1$-sided},
in which an edge can be bundled with other edges only on one of its two
end vertices, that is, each edge has only one fan-bundle; see Figs.~\ref{fig:bundleplanar} and~\ref{fig:k5}.
This restriction immediately implies that $1$-sided \fb drawings are fan-planar. 
Note that this is not the case for the model in which each edge has two fan-bundles, 
which we call \emph{$2$-sided} (see
Fig.~\ref{fig:k6}). In fact, we will prove in Section~\ref{sec:relationships} that there 
exist $2$-sided \fb graphs that are not fan-planar, and vice versa.

Since each bundle collects a set of edges and allows them to participate in a crossing, natural beyond-planarity theoretical questions arise:%
\begin{inparaenum}[(i)]
\item Characterize or recognize the graphs that admit ($1$- or $2$-sided) \fb drawings, and
\item provide upper and lower bounds on their \emph{edge density}, that is, the maximum number of edges they can have, expressed in terms of their number of vertices.
\end{inparaenum}
More graph drawing related questions concern the use of this model embedded in commonly used approaches like hierarchical drawings~\cite{DBLP:journals/tsmc/SugiyamaTT81}, radial drawings~\cite{DBLP:reference/crc/GiacomoDL13}, or force-directed methods~\cite{fr-gdfdp-91}.

We provide several answers to these questions under the $1$-sided and the
$2$-sided models. We study these questions in the general case and in two
restricted variants that have been commonly studied for other classes of
nearly-planar graphs. Namely, in an \emph{outer}-\fb drawing the vertices are
incident to the unbounded face of the drawing, while in a \emph{$2$-layer} \fb
drawing the graph is bipartite and the vertices of the two bipartition sets lie on
two parallel~lines, and the edges lie completely between these two lines.

\myparagraph{Our Contribution.}
\label{sec:contribution}
%
In Section~\ref{sec:relationships}, we study inclusion relationships between
the classes of $1$- and $2$-sided \fb graphs and other classes of nearly-planar 
graphs. Then, in Section~\ref{sec:density}, we present upper and lower
bounds on the edge density of these classes; for an overview refer to Table~\ref{table:density}. 
We then consider the complexity of the recognition
problem; we prove in Section~\ref{sec:npcompleteness} that this problem is
NP-complete in the general case for both the $1$-sided and the $2$-sided models,
while in Section~\ref{sec:recognition} we present linear-time recognition and
drawing algorithms for biconnected 2-layer \fb graphs, maximal 2-layer \fb graphs,
and triconnected outer-\fb graphs in the $1$-sided model.

In Section~\ref{sec:related_work} we present a short overview of the state of the art for~beyond-planarity.  
Section~\ref{sec:preliminaries} introduces preliminary notions and notation.
We conclude in Section~\ref{sec:conclusions} by giving a list of open problems.

\begin{table}
  \centering
  \caption{Lower bounds (LB) and upper bounds (UB) on the number of edges of $1$- and $2$-sided \fb graphs.}
  \label{table:density}
  \medskip
  \begin{tabular}{lc@{\hspace{.15cm}}c@{\hspace{.15cm}}ccc@{\hspace{.15cm}}c@{\hspace{.15cm}}ccc@{\hspace{.15cm}}c@{\hspace{.15cm}}c}
    \toprule
     & \multicolumn{3}{c}{$2$-layer}  & \multicolumn{3}{c}{outer}  & \multicolumn{3}{c}{general}\\
    \cmidrule(r{8pt}){2-4} \cmidrule(r{8pt}){5-7} \cmidrule{8-10}
    Model & LB & UB & Th. &  LB & UB & Th. &  LB & UB & Th.\\
    \midrule
    $1$-sided & $\frac{5n-7}{3}$ & $\frac{5n-7}{3}$ & \ref{thm:density_layered}  & $\frac{8n-13}{3}$ & $\frac{8n-13}{3}$ & \ref{thm:density_outer} &  $\frac{13n-26}{3}$ & $\frac{13n-26}{3}$ & \ref{thm:density_general}\\[4pt]
    $2$-sided & $2n-4$ & $3n-7$ & \ref{thm:2sided-2layer-density} &  $4n-9$ & $4n-9$ & \ref{thm:2sided-outer-density} &  $6n-18$ & $\frac{43n-78}{5}$ & \ref{thm:general_bound}\\
    \bottomrule
  \end{tabular}
\end{table}

\section{Related Work.}
\label{sec:related_work}
%
Over the last few years, several classes of nearly-planar graphs have been
proposed and studied. Apart from \emph{$1$-planar}~\cite{2017arXiv170302261K,MR0187232},
\emph{quasi-planar}~\cite{DBLP:journals/combinatorica/AgarwalAPPS97}, and
\emph{fan-pla\-nar}~\cite{DBLP:journals/corr/KaufmannU14} graphs, which have already
been discussed, other classes of nearly-planar graphs include: %
\begin{enumerate*}[label=(\roman*)]
\item \emph{$k$-planar}~\cite{DBLP:journals/combinatorica/PachT97}, which
  generalize $1$-planar graphs, as they admit drawings in which every edge
  is crossed at most~$k$ times;
\item \emph{fan-crossing free}~\cite{DBLP:conf/isaac/CheongHKK13}, which
  complement fan-planar graphs, as they forbid fan crossings but allow each
  edge to cross any number of pairwise independent edges; 
\item \emph{RAC}~\cite{DBLP:journals/tcs/DidimoEL11}, which admit straight-line
  drawings where edges cross only at right angles; and 
\item the recently introduced \emph{$k$-gap-planar}~\cite{2017arXiv170807653B}, which admit drawings where each crossing is assigned to one of the two involved edges and each edge is assigned at most $k$ of its crossings. 
\end{enumerate*}

These classes have been mainly studied in terms of their edge density, and of
the computational complexity of their corresponding recognition problem. From
the density point of view, while the graphs in all these classes can be denser
than planar graphs, all of them still have a linear number of
edges~\cite{DBLP:journals/combinatorica/AgarwalAPPS97,2017arXiv170807653B,MANA:MANA3211170125,
DBLP:conf/isaac/CheongHKK13,DBLP:journals/tcs/DidimoEL11,
DBLP:journals/corr/KaufmannU14,DBLP:journals/combinatorica/PachT97}. From the
recognition point of view, the problem has been proven to be NP-complete for most of
the classes~\cite{DBLP:journals/jgaa/ArgyriouBS12,2017arXiv170807653B,%
DBLP:journals/tcs/BinucciGDMPST15,DBLP:journals/tcs/EadesHKLSS13}, except for
quasi-planar and fan-crossing free graphs, whose time complexities are still unknown.
On the other hand, for the restricted outer and $2$-layer cases, several
polynomial-time algorithms have been proposed~\cite{DBLP:conf/gd/AuerBBGHNR13,
BekosCGHK14,DBLP:conf/walcom/DehkordiNEH13,eggleton,
DBLP:journals/algorithmica/GiacomoDEL14,DBLP:conf/gd/HongEKLSS13}.

Fink et al.~\cite{fhsv-bceg-LATIN16} considered a different style of edge
bundling, where groups of locally parallel edges are bundled and only bundled
crossings are allowed. 

\section{Preliminaries}
\label{sec:preliminaries}
%
A graph $G$ admitting a $1$-sided ($2$-sided) \fb drawing is called
\emph{$1$-sided} (\emph{$2$-sided}, respectively) \emph{\fb}  graph. Graph $G$ is
\emph{maximal} if the addition of any edge destroys its \fby, in any of its drawings.
Analogously, we define the (maximal) $1$-sided or $2$-sided 
\emph{outer-\fb} and \emph{$2$-layer \fb} graphs. The
drawings we consider in this paper are \emph{almost simple}, meaning that no
two bundles of the same vertex cross. Note however that two edges incident to
the same vertex might cross in the $2$-sided model; see for an example Fig.~\ref{fig:k4-14-2sided}. 
A \emph{rotation system} describes for each vertex an order of its incident edges as they appear around it.

A vertex $u$ can be incident to more than one bundle. Let $B_u$ be one 
such bundle. We say that~$B_u$ \emph{is anchored at} vertex $u$, which is
the \emph{origin} of $B_u$. We denote by $|B_u|$ the \emph{size} of $B_u$,
that is, the number of edges represented by~$B_u$. Clearly, $|B_u| \leq \deg(u)$.
We refer to the endpoint of fan-bundle~$B_u$ different from~$u$ 
(where all the edges of $B_u$ are separated from each other) 
as the \emph{terminal} of $B_u$, and to the end vertex different
from $u$ of any edge in~$B_u$ as a \emph{tip} of $B_u$. 
For two crossing fan-bundles $B_u$ and $B_v$ anchored at vertices $u$ and $v$, we call
\emph{$B_uB_v$-following curve} a curve that starts at~$u$, follows~$B_u$ up to the
crossing point with~$B_v$, then follows~$B_v$, and ends at~$v$ in such a way
that it crosses neither bundles.

\section{Relationships with other graph classes}
\label{sec:relationships}

In this section, we discuss inclusion relationships between the classes of
$1$-sided and $2$-sided \fb graphs and other relevant classes of
nearly-planar graphs. In particular, we focus on the classes of $1$-planar and
fan-planar graphs, due to the immediate relationships determined by the
definition of \fb graphs. We also consider the well-studied class of
$2$-planar graphs~\cite{DBLP:conf/gd/AuerBGH12}, which has already been proven to
be incomparable with the class of fan-planar
graphs~\cite{DBLP:journals/tcs/BinucciGDMPST15}, despite the fact that their maximum edge-density is the
same, namely $5n-10$~\cite{DBLP:journals/corr/KaufmannU14,DBLP:journals/combinatorica/PachT97}.
Our findings are summarized in Fig.~\ref{fig:relationship}.

\begin{figure}[t]
	\centering
	\includegraphics[scale=0.8]{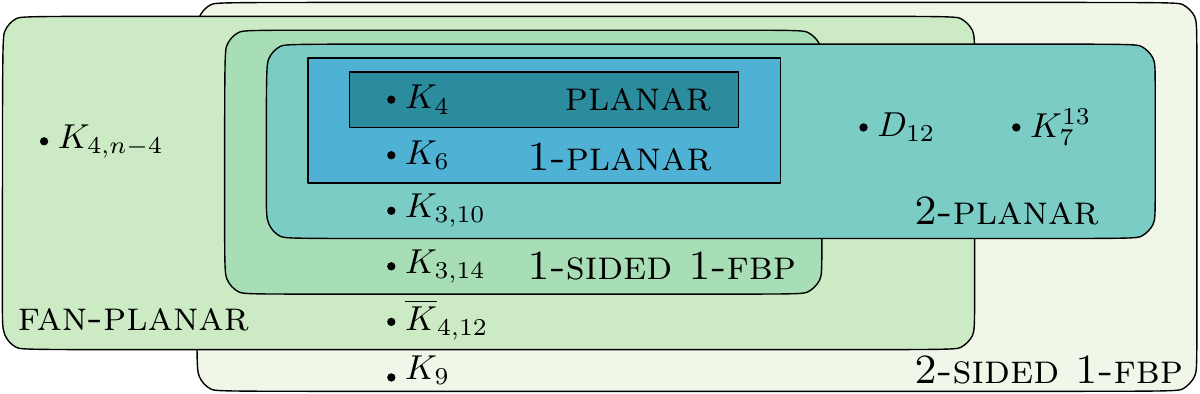}
	\caption{Relationships among graph classes proved in this paper.
	The graph denoted by $\overline{K}_{4,12}$ is obtained from the complete bipartite graph $K_{4,12}$
	by joining the four vertices of its first bipartition set on a path and the twelve vertices of its second bipartition set on a second path 
	(see Fig.~2(a) in~\cite{DBLP:journals/corr/KaufmannU14} or Fig.~\ref{fig:k412}).
	The graph denoted by $D_{12}$ corresponds to the graph obtained from the dodecahedral graph by adding a pentagram in each of its faces
	(see Fig.~2(b) in~\cite{DBLP:journals/corr/KaufmannU14} or Fig.~\ref{fig:d12}). 
	The graph denoted by $K_{7}^{13}$ is obtained by appropriately adjusting 13 copies of $K_7$ (see Fig.~8 in~\cite{DBLP:journals/tcs/BinucciGDMPST15} or Fig.~\ref{fig:kplanar-non-fanplanar}).}
	\label{fig:relationship}
\end{figure}

The inclusion relationship \textsc{$1$-planar} $\subseteq$ \textsc{$1$-sided \fb} $\subseteq$
\textsc{fan-planar} follows from the definition of $1$-sided \fb graphs, and the same holds
for the inclusion relationship \textsc{$2$-planar graphs} $\subseteq$ \textsc{$2$-sided \fb}. Also, Binucci
et al.~\cite{DBLP:journals/tcs/BinucciGDMPST15} proved that the class of $2$-planar 
graphs is incomparable with the class of fan-planar graphs. In particular,
they showed that the 3-partite graph $K_{1,3,10}$ is fan-planar but not 2-planar
and that the graph $K_{7}^{13}$, which is obtained by adjusting 13 copies of $K_7$ as depicted
in Fig.~\ref{fig:kplanar-non-fanplanar}, is 2-planar but not fan-planar.

\begin{figure}[b]
  \centering
  \includegraphics{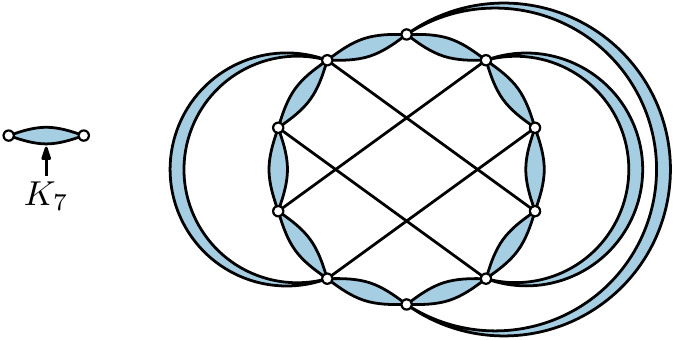}
  \caption{Illustration of the graph $K_7^{13}$ which has the property that it is $2$-planar, but not fan-planar~\cite{DBLP:journals/tcs/BinucciGDMPST15}.}
  \label{fig:kplanar-non-fanplanar}
\end{figure}

We already know that $K_4$ is planar and that $K_5$ and $K_6$ are $1$-planar,
and hence belong to all classes that we consider here.
Kaufmann and Ueckerdt~\cite{DBLP:journals/corr/KaufmannU14} proved that the
graph obtained from the dodecahedral graph by adding a pentagram in
each of its faces (denoted by~$D_{12}$ in Fig.~\ref{fig:relationship}) 
is $2$-planar, fan-planar, and meets exactly the maximum
density of these classes of graphs, namely $5n-10$; see Fig.~2(b)
in~\cite{DBLP:journals/corr/KaufmannU14} or Fig.~\ref{fig:d12}. As we will see in
Section~\ref{sec:density}, this graph is too dense to be $1$-sided \fb (and
hence $1$-planar). Since~$K_9$ contains more than $5n-10$ edges, it is
neither fan-planar nor $2$-planar. On the other hand,~$K_9$ is $2$-sided \fb,
as shown in Fig.~\ref{fig:k9-2sided}. We do not know whether $K_{10}$ is $2$-sided 
\fb or not, but we know that there exists some value of $n$ for which $K_{n}$
is not $2$-sided \fb, since $2$-sided \fb graphs have at most a linear number of
edges, as we prove in Section~\ref{sec:density}. An interesting observation
is that $K_{10}$ admits a quasi-planar drawing~\cite{franz-quasi-planar}.

\begin{figure}[t]
	\centering
    \subfloat[\label{fig:d12}]{\includegraphics[page=2]{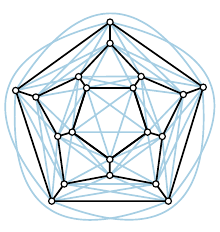}}
    \hfil
    \subfloat[\label{fig:k9-2sided}]{\includegraphics[page=4]{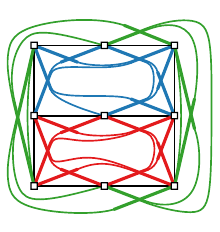}}
    \hfil
    \subfloat[\label{fig:k412-2sided}]{\includegraphics[page=5]{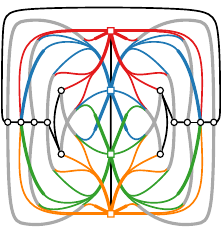}}    

    \subfloat[\label{fig:k412}]{\includegraphics[page=3]{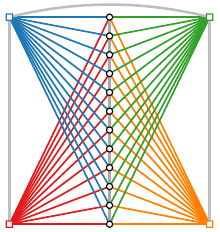}}
    \hfil    
    \subfloat[\label{fig:k314}]{\includegraphics[page=7]{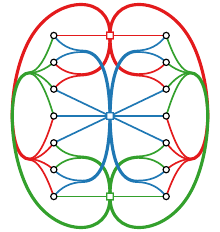}}
    \hfil
    \subfloat[\label{fig:k4-14-2sided}]{\includegraphics[page=6]{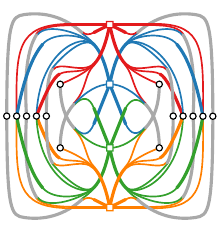}}
    \caption{
	(a)~A straight-line drawing of graph $D_{12}$ obtained from the dodecahedral graph by adding a pentagram in each of its faces.
	(b)~A $2$-sided \fb drawing of $K_9$.
    (c)~A $2$-sided \fb drawing of graph $\overline{K}_{4,12}$ obtained from the complete bipartite graph $K_{4,12}$ by joining
    on a path the four vertices of its first bipartition set and
	on a second path the $12$ vertices of its second bipartition set.
	(d)~A fan-planar drawing of $\overline{K}_{4,12}$.
	(e)~A $1$-sided \fb drawing of $K_{3,14}$.
	(f)~A $2$-sided \fb drawing of~$K_{4,14}$.}
    \label{fig:relationships}
\end{figure}

In Fig.~\ref{fig:k412-2sided}, we show that the graph $\overline{K}_{4,12}$
obtained from the complete bipartite graph $K_{4,12}$ by joining 
the four vertices of its first bipartition set on a path and the twelve
vertices of its second bipartition set on a second path is $2$-sided \fb. 
As shown in Fig.~\ref{fig:k412}, this graph is 
fan-planar~\cite{DBLP:journals/corr/KaufmannU14}, but not $2$-planar (as it contains $K_{3,11}$ as
a subgraph, which is not $2$-planar; see Lemma~\ref{le:K3-n-kplanar}).
In addition, this particular graph cannot be $1$-sided \fb, as it
contains $62$ edges, while a $1$-sided \fb graph on $16$ vertices cannot have
more than $60$ edges (see Section~\ref{sec:density}).

We now give a proof of our previous claim that $K_{3,11}$ is not $2$-planar;
note that even $K_{3,14}$ is $1$-sided \fb, as shown in Fig.~\ref{fig:k314}.
We also show that $K_{3,10}$ is $2$-planar, by means of a more general proof
(which may be of its own interest) about the existence of $k$-planar drawings
of graphs $K_{3,n-3}$ where $n$ is a function of $k$. Note that in the proof of
Lemma~\ref{le:K3-n-kplanar} we assume that edges incident to the same vertex 
are not allowed to cross, which is a reasonable assumption in the study of 
topological graphs, and consequently of $k$-planar graphs.

\begin{lemma}\label{le:K3-n-kplanar}
For each integer $k \geq 0$, graph $K_{3,4k+2}$ is $k$-planar, while graph $K_
{3,4k+3}$ is not $k$-planar.
\end{lemma}
\begin{proof}
For a complete bipartite
graph~$K_{3,n-3}$, let $U=\{u,v,w\}$ be the set of three vertices in the first
bipartition set and let $V$ be the set of $n-3$ vertices in the second
bipartition set. Also, let $E=U \times V$ be the set of its edges.

We show how to obtain a $k$-planar drawing of graph~$K_{3,2k+1}$ such that
the vertices in~$U$ are drawn on the horizontal line $y=0$, the vertices in~$V$
are drawn on the horizontal line with $y=1$, and each edge in~$E$ is drawn
as a curve completely in the half plane above the horizontal line $y=0$; see
Fig.~\ref{fig:kplanar-K3,n}.

\begin{figure}[t]
  \centering
  \subfloat[\label{fig:kplanar-K3,n}]{\includegraphics[page=5]{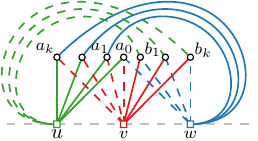}}
  \hfil
  \subfloat[\label{fig:1planar-K3,6}]{\includegraphics[page=2]{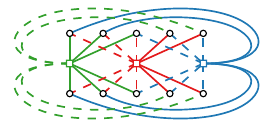}}
  \caption{(a) A $k$-planar drawing of $K_{3,2k+1}$ in one half plane, and
    (b) a 2-planar drawing of~$K_{3,10}$.}
\end{figure}

Denote by $a_0,\ldots,a_k,b_1,\ldots,b_k$ the vertices of the second bipartition set of graph $K_{3,2k+1}$. 
We place vertex~$u$ at point~$(-1,0)$, vertex~$v$ at point~$(0,0)$, and vertex~$w$
at point~$(1,0)$. Then, for each $i=0,\dots,k$, we place vertex~$a_i$ at
point $(-\frac{i}{k},1)$. Symmetrically, for each $j=1,\dots,k$, we
place vertex~$b_j$ at point~$(\frac{j}{k},1)$. We draw the edges of $K_{3,2k+1}$
as follows.

\begin{enumerate}[label=E\arabic*.]
\item Each edge~$(v,a_i)$, for $i=0,\ldots,k$ is drawn as a straight-line segment;
\item each edge~$(v,b_j)$, for $j=1,\ldots,k$ is drawn as a straight-line segment;
\item each edge~$(u,a_i)$, for $i=0,\ldots,k$, is drawn as a straight-line segment;
\item each edge~$(w,b_j)$, for $j=1,\ldots,k$, is drawn as a straight-line segment;
\item edge~$(w,a_0)$ is drawn as a straight-line segment;
\item each edge~$(w,a_i)$, for $i=1,\dots,k$, is drawn as a curve that leaves~$a_i$
  from the top, goes to the right around~$b_k$, and enters~$w$ from the right, in
  such a way that $a_0,b_1,\dots,b_k,a_1,\ldots,a_k$ appear in this clockwise order around~$w$;
\item each edge~$(u,b_j)$, for $j=1,\dots,k$, is drawn as a curve that leaves~$b_j$
  from the top, goes to the left around~$a_k$, and enters~$u$ from the left, in
  such~a way that $a_0,a_1,\ldots,a_k,b_1,\dots,b_k$ appear in this counterclockwise order around~$u$.
\end{enumerate}
That way, we will get the following crossings.
\begin{enumerate}[label=C\arabic*.]
\item Edges~$(u,a_k)$, $(v,a_0)$, and $(w,b_k)$ are drawn crossing-free;
\item every edge~$(v,a_i)$ with $1\le i\le k$ crosses exactly every edge~$(u,a_j)$
  with $0\le j\le i-1$, and thus it has at most~$k$ crossings;
\item every edge~$(u,a_j)$ with $1\le j\le k-1$ crosses exactly every edge~$(v,a_i)$
  with $j+1\le i\le k$, and thus it has at most~$k$ crossings;
\item every edge~$(v,b_i)$ with $1\le i\le k$ crosses exactly every edge~$(w,b_j)$
  with $1\le j\le i-1$, plus the edge~$(w,a_0)$, and thus it has at most~$k$ crossings;
\item every edge~$(w,b_j)$ with $1\le j\le k-1$ crosses exactly every edge~$(v,b_i)$
  with $j+1\le i\le k$, and thus it has at most~$k$ crossings; 
\item edge~$(w,a_0)$ crosses every edge~$(v,b_j)$
  with $1\le j\le k$, and thus it has~$k$ crossings; 
\item every edge~$(w,a_i)$ with $1\le i\le k$ crosses exactly every edge~$(u,b_j)$
  with $1\le j\le k$, and thus it has~$k$ crossings;
\item every edge~$(u,b_j)$ with $1\le j\le k$ crosses exactly every edge~$(w,a_i)$
  with $1\le i\le k$, and thus it has~$k$ crossings.
\end{enumerate}

From the analysis above, it follows that no edge has more than $k$ crossings. Hence, the 
constructed drawing of $K_{3,2k+1}$ is $k$-planar. To
obtain a $k$-planar drawing for graph $K_{3,4k+2}$, we create two copies of
the drawing of $K_{3,2k+1}$, mirror one of them at the horizontal line
$y=0$, and identify the vertices~$u,v,w$ in the two drawings.
Fig.~\ref{fig:1planar-K3,6} shows such a $2$-planar drawing for graph $K_{3,10}$. 
This completes the proof of the first part of the statement.

For the second part of the statement, assume to the contrary that for some $k
\geq 0$ there is a $k$-planar drawing~$\Gamma$ of graph $G=K_{3,4k+3}$. As
above, we denote by $U=\{u,v,w\}$ the first bipartition set of $G$ and by
$V=\{x_1,\ldots,x_{4k+3}\}$ its second bipartition set.

We will first show that there is (at least) one subgraph~$G'=K_{2,2}$ of $G$
such that the drawing~$\Gamma'$ of~$G'$ contained in~$\Gamma$ is planar. For
each $j=2,\ldots, 4k+3$, let $G_j$ be the subgraph of $G$ induced by vertices
$u$ and $v$ of the first bipartition set of $G$ and by vertices $x_1$ and $x_j$
of its second bipartition set. Note that all subgraphs $G_2,\ldots,G_{4k+3}$ have
the edges $(u,x_1)$ and $(v,x_1)$ in common, which do not cross each other
in~$\Gamma$, because they are both incident to $x_1$. Analogously, the edges~$(u,x_j)$ and~$(v,x_j)$, with
$2\le j \le 4k+3$, do not cross each other. Thus, in any subgraph~$G_j$, we
can only have a crossing between~$(u,x_1)$ and~$(v,x_j)$, or between~$(v,x_1)$
and~$(u,x_j)$. Since~$\Gamma$ is a $k$-planar drawing, the edges~$(u,x_1)$
and~$(v,x_1)$ can only be crossed~$2k$ times in total. Hence, at least~$2k+2$
of the subgraphs from~$G_2,\ldots,G_{4k+3}$ (i.e.~excluding $G_1$) induce a crossing-free drawing.

Assume w.l.o.g.\ that the drawing~$\Gamma'$ of~$G'=G_2$ contained in~$\Gamma$
is planar. Note that $\Gamma'$ consists of a closed simple curve
through~$u,x_1,v,x_2$; see Fig.~\ref{fig:kplanar-interior-1}. We denote by $\mathcal I$
the bounded region enclosed by this curve, and by~$\mathcal O$ the unbounded
region outside this curve. We will now show that~$w$ can lie neither
in~$\mathcal I$ nor in~$\mathcal O$.

Suppose that~$w$ lies in the bounded region~$\mathcal I$. The case where $w$
lies in the unbounded region~$\mathcal O$ is symmetric. Since~$G'$ has
exactly four edges and~$\Gamma$ is a $k$-planar drawing, the edges of~$G'$
can be involved in at most~$4k$ crossings in total. Since~$w$ is adjacent to
all vertices~$x_3,\ldots,x_{4k+3}$, at least one of them, say $x_3$, has to
lie inside~$\mathcal I$, as in Fig.~\ref{fig:kplanar-interior-1}.
If both edges~$(u,x_3)$ and~$(v,x_3)$ are drawn completely inside~$\mathcal I$,
then they split~$\mathcal I$ into two bounded regions~$\mathcal I'$
and~$\mathcal I''$, delimited by the closed curves through $u,x_1,v,x_3$ and
through $u,x_3,v,x_2$, respectively; see Fig.~\ref{fig:kplanar-interior-2}.
In this case, however, we could apply the same argument as above to say that
there exists at least a vertex that lies in the interior of the same bounded
region as $w$, either $\mathcal I'$ or $\mathcal I''$. Note that the bounded
region we consider in this step is smaller and contains fewer vertices than
$\mathcal I$. Thus, by repeating this argument at most a linear number of
times, we can prove that there exists a vertex $x_j$, with $3 \leq j \leq 4k+3$,
lying inside the same bounded region $\mathcal{I}^*$ as $w$, such that one
of the edges $(u,x_j)$ and $(v,x_j)$ crosses an edge of the graph $G^*=K_{2,2}$
delimiting $\mathcal{I}^*$. To simplify the notation, assume $j=3$, $G^*=G'$,
and $\mathcal{I}^*=\mathcal{I}$; see Fig.~\ref{fig:kplanar-interior-3}.

\begin{figure}[tb]
  \centering
  \subfloat[\label{fig:kplanar-interior-1}]{\includegraphics[page=1]{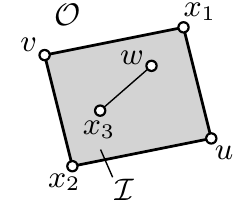}}
  \hfil
  \subfloat[\label{fig:kplanar-interior-2}]{\includegraphics[page=2]{kplanar-interior}}
  \hfil
  \subfloat[\label{fig:kplanar-interior-3}]{\includegraphics[page=3]{kplanar-interior}}
  \caption{Illustration of the proof of Lemma~\ref{le:K3-n-kplanar}:
  (a)~$w$ lies inside~$\mathcal I$,
  (b)~$(u,x_3)$ and $(v,x_3)$ are planar, and
  (c)~$(u,x_3)$ or $(v,x_3)$ is crossed}
  \label{fig:kplanar-interior}
\end{figure}

Hence, there are at most $4k-1$ crossings between the edges of $G'$ and edges
not incident to~$x_3$. This implies that another one of the remaining~$4k$
vertices $x_4,\ldots,x_{4k+3}$ must lie inside $\mathcal I$. By iteratively
applying this argument, we conclude that all vertices $x_3,\ldots,x_{4k+3}$
have to lie inside~$\mathcal I$ and that, for every~$i=3,\ldots,4k+3$, at
least one of the edges~$(u,x_i)$ and~$(v,x_i)$ has to cross an edge of~$G'$.
However, this implies that the~four edges of~$G'$ are involved in at
least~$4k+1$ crossings in total; a contradiction. Thus, $K_{3,4k+3}$ has no
$k$-planar drawing.
\end{proof}

As already mentioned, the complete bipartite graph $K_{4,n-4}$ is fan-planar for
every~$n \geq 4$. In the following we will prove that there exists a value
of~$n$ such that $K_{4,n-4}$ is not $2$-sided \fb, which also proves that fan-planar
 (and hence quasi-planar) graphs do not form a subclass of $2$-sided \fb
graphs. Note that a $2$-sided \fb drawing can be constructed for $K_{4,14}$;
see Fig.~\ref{fig:k4-14-2sided}. Before we proceed with the detailed proof
of our claim, we need two auxiliary lemmas.

\begin{lemma}\label{lem:not-many-common}
  Let $\Gamma$ be a $2$-sided \fb drawing of graph $K_{3,9}$. Then, at least
  two of the vertices of the first bipartition set of $K_{3,9}$ must use more than one fan-bundle in $\Gamma$.
\end{lemma}
\begin{proof}
Denote by $U=\{u,v,w\}$ the first bipartition set of $K_{3,9}$ and by
$V=\{x_1,\ldots,x_9\}$ its second bipartition set. Consider the graph $K_{2,9}=\{u,v\} \times V$ 
and assume that there exists a drawing $\Gamma$ of this graph in which $u$ and $v$
are connected to all of $x_1,\dots,x_9$ by means of single fan-bundles $B_u$
and $B_v$, respectively. In other words, there exists a fan-bundle $B_u$ (a fan-bundle $B_v$) anchored at $u$ (at~$v$) whose tips are $x_1,\dots,x_9$.

Consider two edges $(u,x_i)$ and $(v,x_j)$, with $1 \leq i \neq j \leq 9$;
if these edges cross each other in $\Gamma$, we say
that they form an \emph{intersecting pair}; see the two pairs
$\langle (u,x_1),(v,x_2)\rangle$ and $\langle (u,x_3),(v,x_4)\rangle$ in
Fig.~\ref{fig:pair-intersect}. We show that, whatever is the number of intersecting pairs, it is
not possible to add the third vertex $w$ to~$\Gamma$ and connect it to all
vertices $x_1,\dots,x_9$, so to obtain a $2$-sided \fb drawing of~$K_{3,9}$.

The proof is based on the following observation. For any two edges forming an
intersecting pair $\langle (u,x_i),(v,x_j)\rangle$, consider the curve connecting the
terminals of $B_u$ and $B_v$, constructed as follows. First follow $(u,x_i)$ from the terminal
of~$B_u$ till the intersection point of the bundles anchored at $x_i$ and $x_j$ containing $(u,x_i)$ and $(v,x_j)$, respectively; then follow $(v,x_j)$ till the
terminal of $B_v$; see the black dotted lines in Fig.~\ref{fig:pair-intersect}.
By construction, this curve is not crossed by any edge in~$\Gamma$ that is not incident
to either $x_i$ or $x_j$. This already implies that there exist no three intersecting pairs. In this
case, in fact, there exists no placement for~$w$ that allows us to connect it to
all the vertices $x_1,\dots,x_9$ (and in particular to the at least six
vertices among $x_1,\dots,x_9$ involved in the intersecting pairs) without crossing at least one of
such curves.

\begin{figure}[t]
	\centering
	\includegraphics[scale=1.1]{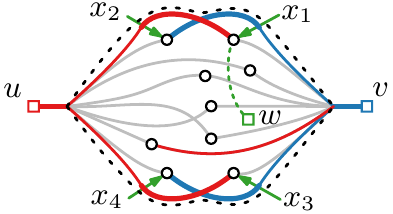}
	\caption{Illustration for the proof of Lemma~\ref{lem:not-many-common}, when
    there exist two intersecting pairs $\langle (u,x_1),(v,x_2)\rangle$ and
    $\langle (u,x_3),(v,x_4)\rangle$.}
	\label{fig:pair-intersect}
\end{figure}

Suppose now that there exist exactly two intersecting pairs, as in
Fig~\ref{fig:pair-intersect}. W.l.o.g., let
$\langle (u,x_1), (v,x_2)\rangle$ and $\langle (u,x_3),(v,x_4)\rangle$ be these pairs. Since the two
curves defined by these two pairs cannot be crossed in $\Gamma$, there exists
one of the two regions, say $R$, delimited by these curves that contains all
the vertices $x_1,\dots,x_9$, as well as vertex $w$, in its interior.

Consider now the five vertices $x_5,\dots,x_9$ and the five paths between $u$
and $v$ passing through these vertices. Observe that the edges of these paths
may cross each other, but the only crossings can be either between two edges
incident to~$u$ or between two edges incident to $v$, as otherwise there
would be an additional intersecting pair. Consider the subregions of $R$
defined by the arrangement of the curves representing these paths. Note that
these subregions are at least six, which happens when all the five paths are
crossing-free. By the previous observation on the possible crossings between
these paths, we can conclude that, if we place $w$ in any of these subregions,
either the edges connecting $w$ to $x_1$ and $x_2$, or those connecting $w$
to $x_3$ and $x_4$ have to cross edges of at least three of the five paths;
see Fig.~\ref{fig:pair-intersect}. This is not possible, since the edges of
two different paths cannot be bundled together and since any edge incident to
$w$ has only two fan-bundles.

The case in which there exists at most one intersecting pair is analogous. In
fact, in this case we can consider the region $R$ as the whole plane, and use
the at least seven paths not involved in the intersecting pair to make the
same argument as above. This concludes the proof of the lemma.
\end{proof}

\begin{lemma}\label{lem:many-terminals}
  In any $2$-sided \fb drawing of graph $K_{3,n-3}$, there is a fan-bundle
  containing at least $(n-2)/8$ edge-segments.
\end{lemma}
\begin{proof}
  By Lemma~\ref{le:K3-n-kplanar}, graph $K_{3,n-3}$ is not $k$-planar for
  $k \leq (n-6)/4$. This implies that in any drawing of $K_{3,n-3}$ there is at
  least one edge with at least $1+(n-6)/4=(n-2)/4$ crossings. Since in a
  $2$-sided \fb drawing every edge has crossings only at its two fan-bundles,
  there is one of such fan-bundles that crosses at least $(n-2)/8$ edges. Since
  all these edges must be bundled together, the statement follows.
\end{proof}
\noindent We are now ready to give the detailed proof of our initial claim, i.e., that there exists a value of~$n$ such that~$K_{4,n-4}$ is not $2$-sided \fb.
\begin{theorem}
  Graph $K_{4,n-4}$ is not $2$-sided \fb for $n \geq 571$.
\end{theorem}
\begin{proof}
  Assume to the contrary that $K_{4,n-4}$ admits a $2$-sided \fb drawing~$\Gamma$
  for some $n\geq 571$. First, consider any subgraph $K_{3,n-4}$ of $K_{4,n-4}$.
  By Lemma~\ref{lem:many-terminals}, there exists a fan-bundle~$B_u$
  anchored at a vertex~$u$ with at least $(n-3)/8$ edge-segments. Since all the
  vertices in the second bipartition set have degree~$3$, vertex~$u$ belongs to the
  first bipartition set.
  Consider now the subgraph $K_{3,(n-3)/8}$ of $K_{4,n-4}$ that is composed of
  the three vertices of the first bipartition different from~$u$ and from
  the~$(n-3)/8$ vertices that are connected to~$u$ by means of~$B_u$. By
  Lemma~\ref{lem:many-terminals}, there exists a fan-bundle~$B_v$ anchored at a
  vertex~$v$ with at least $\frac{1+(n-3)/8}{8}=(n+5)/64$ edge-segments. Again,
  vertex~$v$ belongs to the first bipartition set.

  We have now found $(n+5)/64$ vertices that are connected to~$u$ and to~$v$ by
  means of single fan-bundles~$B_u$ and~$B_v$ in~$\Gamma$. For $n \geq 571$,
  this value is at least $9$, which contradicts Lemma~\ref{lem:not-many-common}.
  Hence, the proof of the theorem follows.
\end{proof}

\section{Density}
\label{sec:density}
In the following section, we consider Tur\'an-type problems for \fb graphs. 
Namely, we ask what is the maximum number of edges that an $n$-vertex $1$-sided or $2$-sided \fb graph can have. 
We provide answers to this question both in the general case of the problem and in the outer and $2$-layer variants. 

\subsection{Density of $1$-sided \fblong graphs}
\label{subsec:densityonesided}

We first study the density of general $1$-sided \fb graphs, and then we appropriately adjust for the outer and $2$-layer variants.

\begin{theorem}\label{thm:density_general}
	A $1$-sided \fb graph with $n \geq 3$ vertices has at most $(13n-26)/3$
  edges, which is a tight bound for infinitely many values of $n$.
\end{theorem}
\begin{proof}
  Let $\Gamma$ be a $1$-sided \fb drawing of a \emph{maximally dense} $1$-sided
  \fb graph~$G$ with $n$ vertices, namely, a graph of this class with the
  largest possible number of edges. To estimate the maximum number of edges
  that $G$ may contain, we will first appropriately transform $G$ into a (not
  necessarily simple) maximal planar graph that contains no pairs of \emph{
  homotopic parallel edges}, i.e., both the interior and the exterior regions
  defined by any pair of parallel edges contain at least one vertex of $G$.
  Note that under this assumption, the maximum number of edges of a planar
  multi-graph with $n$ vertices is still $3n - 6$. Since $G$ is a drawn graph, we say
  that $\Gamma$ \emph{contains} an edge $e$ if there exists a drawn edge of $G$
  in $\Gamma$ that is homotopic to $e$.

  Consider two crossing fan-bundles $B_u$ and $B_v$ in $\Gamma$ anchored at
  vertices~$u$ and $v$ of $G$, respectively; see Fig.~\ref{fig:onebundleperedge1}.
  Let $(u,u_1),\ldots,(u,u_\mu)$ and $(v,v_1),\ldots,(v,v_\nu)$ be the edges
  that are bundled in $B_u$ and $B_v$, respectively, in the order in which they
  appear around the terminals of $B_u$ and $B_v$ in $\Gamma$, such that $(u,u_1)$
  and $(v,v_1)$ are the edges that follow $B_u$ and $B_v$ along their
  terminals in clockwise direction. Note that the tips of fan-bundle $B_u$ are not
  necessarily distinct from the tips of fan-bundle $B_v$, that is, for some pairs of
  indices $i$ and $j$ with $1 \leq i \leq \mu$ and $1 \leq j \leq \nu$ it may
  hold that $u_i = v_j$ (e.g., in Fig.~\ref{fig:onebundleperedge1} the third tip
  of $B_u$ is identified with the fourth tip of $B_v$, or in other words $u_3=v_4$).  
  In particular, $v_1 = u_\mu$ holds by the maximality of $G$.

	\begin{figure}[t]
		\centering
		\subfloat[\label{fig:onebundleperedge1}{}]
		{\includegraphics[page=1]{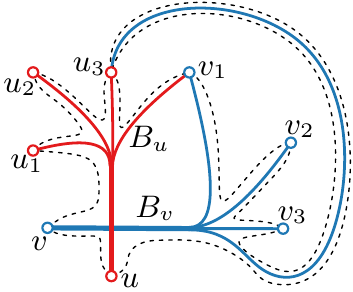}}
		\hfil
		\subfloat[\label{fig:onebundleperedge2}{}]
		{\includegraphics[page=2]{onebundleperedge}}
		\caption{%
		Illustration of the transformation used in the proof of Theorem~\ref{thm:density_general} for the case $\mu=\nu=4$.
		(a)~Two crossing fan-bundles; observe that $v_1=u_4$ and $u_3=v_4$. 
		(b)~Our transformation which leads to two non-homotopic copies of $(u,u_3)$.}
		\label{fig:onebundleperedge}
	\end{figure}

  Consider the edge $(u,v)$ that one can draw in $\Gamma$ as a
  $B_uB_v$-following curve. We refer to this particular edge as the
  \emph{base-edge} of $B_u$ and $B_v$. Since $G$ is maximally dense, graph
  $G$ must contain this edge, as otherwise one could add it in $\Gamma$ without
  violating its $1$-sided \fby or its non-homotopy. Similarly, one can prove
  that~$\Gamma$ contains edges $(v,u_1)$, $(u_1,u_2)$, $\ldots$,
  $(u_{\mu-1},u_{\mu})$, $(u_{\mu},v_1)$, $(v_1,v_2)$, $\ldots$,
  $(v_{\nu-1},v_\nu)$, and $(v_\nu,u)$ that
  are drawn very close either to $B_u$ and $B_v$ or to the unbundled parts of
  the edges incident to $u$ and $v$ (refer to the dotted drawn edges of Fig.~\ref{fig:onebundleperedge1}).
  
  We proceed by applying a simple transformation; see
  Fig.~\ref{fig:onebundleperedge2}. We remove from~$G$ all edges bundled in $B_v$
  and we introduce edges $(u,v_2)$, $\ldots$, $(u,v_{\nu-1})$ drawn
  without crossings in the interior of the region defined by edges
  $(u,v_1)$, $(v_1,v_2)$, $\ldots$, $(v_{\nu-1},v_\nu)$, and $(v_\nu,u)$ in $\Gamma$. Observe
  that this simple transformation does not introduce homotopic parallel edges
  and simultaneously eliminates the crossing between~$B_u$ and $B_v$. However,
  since $\Gamma$ contains edges $(v,v_1)$ and $(v,v_\nu)$, which are not
  contained in the transformed drawing, our transformation leads to a reduction
  by at most two edges. In particular, it is not difficult to see that%
  \begin{inparaenum}[(C.1)]
  \item \label{c:1-planar} if $\nu \neq 1$ and $\mu \neq 1$, then our transformation 
  leads to a reduction by exactly two edges;
  \item \label{c:non-1-planar} otherwise, it leads to a reduction by one edge.   
  \end{inparaenum}
  
  By applying this transformation recursively to every pair of
  crossing fan-bundles, we will obtain a planar drawing $\Gamma'$ of
  a (not necessarily simple) graph~$G'$ on the same set of vertices
  as $G$ that contains no pairs of homotopic parallel edges. Hence,
  graph $G'$ has at most $3n-6$ edges and at most $2n-4$ faces.
  Note that since the edges that are affected by a transformation are
  delimited by uncrossed edges, it follows that they will not be
  ``destroyed'' by another transformation later in our recursive
  procedure.
  
  Now, observe that each transformation of Case~C.\ref*{c:1-planar} produces at least
  three faces in $G'$ and leads to a reduction by exactly two edges
	(see the colored faces of Fig.~\ref{fig:onebundleperedge2}). 
  On the other hand, each transformation of Case~C.\ref*{c:non-1-planar}
  produces two faces in $G'$ and as already mentioned leads to a
  reduction by exactly one edge.
  Hence, graph $G$ contains as many edges as $G'$ plus twice the number of
  transformation of Case~C.\ref*{c:1-planar}, plus the number
  of transformations of Case~C.\ref*{c:non-1-planar}. 
  Let~$f'$ and $f''$ be the number of faces of~$\Gamma'$ created by 
  transformations of Case~C.\ref*{c:1-planar} and C.\ref*{c:non-1-planar}, respectively.
  It follows that~$f'+f''\le 2n-4$. 
  Thus,~$G$ has at most 
  $3n-6 + 2\cdot\lfloor f'/3\rfloor+ f''/2 \leq 3n-6+ 2\cdot\lfloor (2n-4)/3\rfloor \leq (13n-26)/3$~edges.
   
  To show that this upper bound is tight, consider a planar graph
  $\mathcal{P}_{n}$ on $n$ vertices whose faces are all of length~$5$. Hence, $n$ must be appropriately chosen. By
  Euler's formula, graph $\mathcal{P}_{n}$ has exactly $(5n-10)/3$ edges and $(2n-4)/3$ faces.
  In each face, it is possible to add four edges without violating \fby (see,
  e.g.,~Fig.~\ref{fig:k5}); thus, the resulting graph has
  $(5n-10)/3 + 4 \cdot (2n-4)/3 =(13n-26)/3$ edges in total, and the statement follows.
\end{proof}

The next two theorems present tight bounds for the density of $1$-sided \fb
graphs in the outer and in the $2$-layer models. The proofs are based on the
same technique used in the proof of Theorem~\ref{thm:density_general}.

\begin{theorem}\label{thm:density_outer}
  A $1$-sided outer-\fb graph with $n \geq 5$ vertices has at most
  $(8n-13)/3$ edges, which is a tight bound for infinitely many values of $n$.
\end{theorem}
\begin{proof}
  Let $\Gamma$ be a $1$-sided outer-\fb drawing of a maximally dense
  $1$-sided outer-\fb graph $G$ with $n$ vertices. We apply the same
  transformation that we applied in the proof of
  Theorem~\ref{thm:density_general}. In this case, the resulting
  drawing $\Gamma'$ is the drawing of an outerplanar graph $G'$ on
  the same set of vertices as~$G$. Hence, graph $G'$ has at most
  $2n-3$ edges and $n-2$ internal faces. Using the same sequence of arguments as
  in the proof of Theorem~\ref{thm:density_general}, it follows that
  $G$ has at most $2n-3+ 2\cdot\lfloor (n-2)/3 \rfloor \leq
  (8n-13)/3$ edges.
  
  To show that this bound
  is tight, consider an outerplanar graph $\mathcal{O}_{n}$ on $n$ vertices
  whose internal faces are all of length~$5$. By Euler's formula, graph
  $\mathcal{O}_{n}$ has $(4n-5)/3$ edges and $(n-2)/3$ internal faces. Since in each internal
  face of $\mathcal{O}_{n}$ it is possible to add four edges without violating
  outer-\fby (see, e.g.,~Fig.~\ref{fig:k5}), the resulting graph has
  $(4n-5)/3 + 4\cdot (n-2)/3 =(8n-13)/3$ edges in total, and the statement follows.
\end{proof}

\begin{theorem}\label{thm:density_layered}
  A $1$-sided $2$-layer \fb graph with $n \geq 5$ vertices has at
  most $(5n-7)/3$ edges, which is a tight bound for infinitely many values of $n$.
\end{theorem}
\begin{proof}
  Let $G$ be a $1$-sided $2$-layer \fb graph with $n$ vertices. One can add $n-2$
  edges in $G$ to connect in two paths the vertices of each bipartition set and
  obtain a new graph~$G'$ that is $1$-sided outer-\fb. Since by
  Theorem~\ref{thm:density_outer} graph $G'$ cannot have more than $(8n-13)/3$ edges, it
  follows that $G$ cannot have more than $(8n-13)/3-(n-2)=(5n-7)/3$ edges. 
  
\begin{figure}[t]
	\centering
	\includegraphics{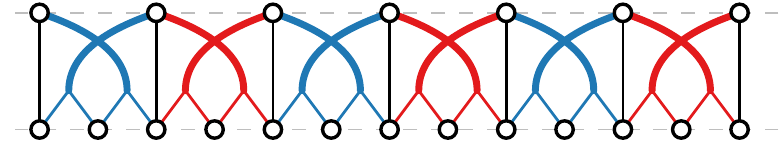}
	\caption{Illustration for the proof of Theorem~\ref{thm:density_layered}.}
	\label{fig:density_layered}
\end{figure}    
  
  A graph $\mathcal{B}_{n}$ with $n$ vertices meeting exactly this
  bound can be easily constructed as follows. Let $n=3k+2$ for some
  positive integer $k$.
  Graph $\mathcal{B}_{n}$ has $k+1$ vertices in its first bipartition set
  and $2k+1$ vertices in its second bipartition set. For each
  $j=1,2,\ldots,k$ the vertices of the first bipartition set of
  $\mathcal{B}_{n}$ with indices~$j$ and $(j+1)$ form a $K_{2,3}$
  with the vertices of the second bipartition set of $\mathcal{B}_{n}$
  with indices $(2j-1)$, $2j$ and $(2j+1)$; see
  Fig.~\ref{fig:density_layered}. Graph $\mathcal{B}_{n}$ is indeed
  $1$-sided $2$-layer \fb, as it consists of $k$ consecutive copies
  of $K_{2,3}$ (which is a $1$-sided $2$-layer \fb graph). To see that
  $\mathcal{B}_{n}$ meets the density bound of this theorem observe that
  each of the $k$ copies of $K_{2,3}$ contributes $6$ edges
  in graph $\mathcal{B}_{n}$ and that consecutive copies of $K_{2,3}$ share
  an edge. Hence, graph $\mathcal{B}_{n}$ has $6k-(k-1)=5k+1$ edges in total,
  and the statement follows.
\end{proof}

\subsection{Density of $2$-sided \fblong graphs}
\label{subsec:densitytwosided}

As opposed to $1$-sided \fb graphs, $2$-sided \fb graphs are not necessarily fan-planar.
In fact, they can be significantly denser than fan-planar graphs, as we will see later in this section.
In the following, we will first present a tight bound on the density of $2$-sided outer-\fb graphs.
We then establish upper and lower bounds for the density of $2$-sided
$2$-layer and $2$-sided general \fb graphs.

We start by presenting a family of $2$-sided outer-\fb graphs with $n$ vertices and $4n-9$
edges. Before doing so, we introduce two definitions. A \emph{flower drawing} of a graph with vertex-set $\{v_1,\dots,v_n\}$ is 
a $2$-sided outer-\fb drawing in which%
\begin{inparaenum}[(i)]
	\item the vertices $v_1,\dots,v_n$ lie on a circle $\mathcal{C}$ in this
    clockwise order,
	\item each vertex $v_i$ has two fan-bundles, a \emph{right} and a \emph{left}
    one; the right (left) fan-bundle is the one that is encountered first (second) when moving clockwise around $v_i$, starting from the outer face (in other words, when seen from the center of $\mathcal{C}$ the left fan-bundle is to the left of the right fan-bundle), and
	\item for each $i=1,\dots,n$, the right fan-bundle of $v_i$ crosses the
    left fan-bundle of $v_{i+1}$ (we assume the indices to be $\bmod{~n}$); see Fig.~\ref{fig:waterlily}.
\end{inparaenum}

A \emph{water lily} is a flower drawing of a graph with $n \geq 9$ vertices
where the terminals of the fan-bundles are partitioned into three sets $S_1$,
$S_2$, and $S_3$, such that %
\begin{inparaenum}[(i)]
	\item each set $S_j$, for $j=1,2,3$, contains at least seven consecutive terminals (e.g., the terminals contained in $S_1$ in Fig.~\ref{fig:waterlily} are the ones spanned by the orange arc),
	\item each pair of sets $S_j$ and $S_k$, with $j \neq k$, have one terminal in common (refer to the terminals that are pointed by the arrows of the arcs indicating $S_1$, $S_2$ and $S_3$ in Fig.~\ref{fig:waterlily}),
    which belongs to the right fan-bundle of a vertex,
	\item the terminal of the right fan-bundle of each vertex $v_i$ is connected to
    the terminal of the left fan-bundles of vertices $v_{i+1}$ and~$v_{i+2}$, and
	\item the terminals in each set $S_j$, 
    for $j=1,2,3$, are connected by a \emph{zigzag-pattern} such that all but two faces have
    degree~3, the other two have degree~4 in order to avoid
    parallel edges; see Fig.~\ref{fig:waterlily}.
\end{inparaenum}

\begin{figure}[t]
	\centering
	\subfloat[\label{fig:waterlily}]
	{\includegraphics[page=7]{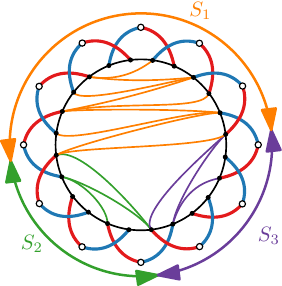}}
	\hfil
	\subfloat[\label{fig:k8}]
	{\includegraphics[page=4]{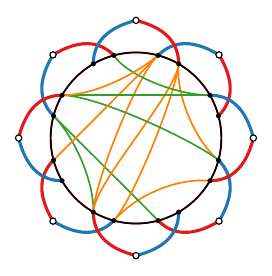}}
  
	\subfloat[\label{fig:doublewaterlily}]
	{\includegraphics[page=3]{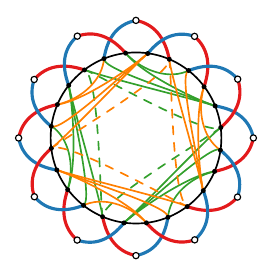}}
	\hfil
	\subfloat[\label{fig:middlebundles}]
	{\includegraphics[page=2]{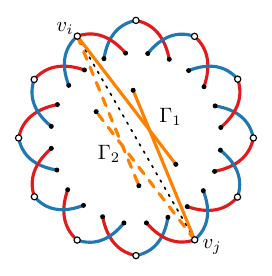}}
	\caption{Illustration of:
		(a)~a water lily,
		(b)~a $2$-sided \fb drawing of $K_8$,
		(c)~a $2$-sided \fb drawing of a graph with $n$ vertices and $6n-18$ edges for $n=12$, and
		(d)~crossing middle fan-bundles.
		In (b) and (c) the orange edges can be drawn on the outer face of the drawing by using twice as many fan-bundles.}
	\label{fig:samples}
\end{figure}

\begin{lemma}\label{lem:waterlily}
There exist $2$-sided outer-\fb graphs with $n$ vertices and exactly $4n-9$ edges, where $n \geq 9$.
\end{lemma}
\begin{proof}
	We prove the statement by showing that in a water lily on $n \geq 9$
  vertices there exist $4n-9$ edges. Namely, consider the graph $H$ whose
  vertices are the terminals of the fan-bundles and whose edges are the
  unbundled parts of the edges of the water lily (drawn plain in
  Fig.~\ref{fig:waterlily}). First, observe that $H$ has $2n$ vertices, since each original
  vertex has one left and one right fan-bundle. Also, $H$ is biconnected and
  outerplanar, by the construction of the water lily. Finally, all the internal
  faces of $H$ are triangular, except for six faces (two for each set $S_j$),
  each of which has degree~$4$. Since a biconnected outerplanar
  graph on $\nu$ vertices in which every internal faces is triangular has $2\nu-3$
  edges, we have that $H$ has $2\cdot(2n)-3-6=4n-9$ edges, and the statement follows.
\end{proof}

Before we proceed with the proof of our upper bound on the edge
density of $2$-sided outer-\fb graphs, we first make a useful
observation.
If the vertices of the graph do not necessarily have to lie on the outer face 
of the drawing, one can draw another set of $2n-9$ edges on the outer face of
a water lily and obtain a $2$-sided \fb drawing of a graph with $6n-18$
edges. Note that this construction corresponds to merging two copies of a
water lily on vertices $v_1,\dots,v_n$ by identifying their vertices and by
keeping only one copy of each edge $(v_i,v_{i+1})$ and $(v_i,v_{i+2})$. With a little effort we
can avoid potential parallel edges that may appear when merging the two
copies. Let $S_1,S_2,S_3$ and $S_1', S_2',S_3'$ be the partitions of the
terminals of the two copies. We require that the terminal shared by $S_j'$
and $S_{j+1}'$ belongs to $S_j$, for each $j=1,2,3$ (where the indices are $\bmod{~3}$);
see Fig.~\ref{fig:doublewaterlily}. In this way, the zigzag patterns on
$S_1,S_2,S_3$ and $S_1', S_2',S_3'$ are edge-disjoint, as desired. We
summarize this observation in the following lemma.

\begin{lemma}\label{lem:2sided-density}
	There exist $2$-sided \fb graphs with $n$ vertices and exactly $6n-18$ edges, where $n \geq 9$.
\end{lemma}

In the following theorem, we show that a $2$-sided outer-\fb graph with $n$ vertices 
cannot have more edges than a corresponding water lily with $n$ vertices. This immediately 
implies that the bound of $4n-9$ is tight for the edge density of $2$-sided outer-\fb graphs. 

Consider a $2$-sided outer-\fb graph $G$ together with a corresponding drawing $\Gamma$.
Let $v_{1},\ldots,v_{n}$ be the vertices of $G$ as they appear in
clockwise order along the outer face of $\Gamma$.
A vertex of $G$ may be incident to several fan-bundles in $\Gamma$.
We denote by~$B_i^1 \ldots B_i^{\lambda}$ the fan-bundles anchored at a vertex $v_i$, as
they appear in clockwise order around $v_i$ starting from the outer face.
We call $B_i^1$ and $B_i^{\lambda}$ the \emph{right fan-bundle} and the
\emph{left fan-bundle} of $v_i$, respectively; we refer to the remaining fan-bundles
anchored at $v_i$ as \emph{middle fan-bundles}. We assume w.l.o.g.\ that the right
fan-bundle of $v_i$ crosses the left fan-bundle of $v_{i+1}$, and that the edge $(v_i,v_{i+1})$
is  represented with a crossing-free unbundled part connecting their terminals.
Note that indeed this is an assumption w.l.o.g.\ as $v_i$ and $v_{i+1}$ are
consecutive along the outer face, and hence representing the edge $(v_i,v_{i+1})$
in this way does not forbid the presence of any other edge.

For our proof we will need two auxiliary lemmas, which consider a special case.
Namely, every vertex of $G$ has degree at least 5 and it has at most
one middle fan-bundle; also if a vertex has a middle fan-bundle,
then it is not involved in any crossing.

As in the proof of Lemma~\ref{lem:waterlily}, consider the graph $H$ whose
vertices are the terminals of all the left, middle, and right fan-bundles
in~$\Gamma$ and whose edges are the unbundled parts of the edges between such
terminals. If we denote by $k$ the number of vertices of $G$ that have a middle fan-bundle,
then $H$ has $2n+k$ vertices. Since $H$ is clearly outerplanar,
it has at most $4n+2k-3$ edges.
We refer to the edges that connect consecutive terminals along the outerface of $H$ as \emph{outer} edges. 
The remaining edges of $H$ are referred to as \emph{inner} edges.

\begin{lemma}\label{cl:outer-face}
	Graph $H$ has at most $2n-k$ outer edges.
\end{lemma}
\begin{proof}
  Since $H$ is outerplanar on $2n+k$ vertices, it cannot have more than $2n+k$ outer edges.
  Consider now a vertex $v_i$ of $G$ with a middle fan-bundle in~$\Gamma$. 
  Then, the vertex of $H$ corresponding to the terminal of the right fan-bundle of $v_{i-1}$
  lies between the vertices of $H$ corresponding to the terminals of the left and of the middle
  fan-bundles of $v_i$ along the outer face of $H$, which implies that there exist
  two outer edges of $H$ that represent the same edge $(v_i,v_{i-1})$
  of $G$. Since $G$ is simple, only one of these outer edges can belong to $H$. For the same reason, there
  exist two outer edges of $H$ that represent the same edge $(v_i,v_{i+1})$ of $G$. Since
  there exist in total $k$ vertices with a middle fan-bundle, the lemma follows.
\end{proof}

In the following lemma, we give also an upper bound on the number of inner edges of $H$.
We do so assuming that all the $2n+k$ outer edges of $H$ are
present (i.e., even the ones that correspond to parallel edges
in $G$); in other words, if an edge of $G$ can be represented both
as an outer and as an inner edge of $H$, then we choose the former
option. Note that this assumption is without loss of generality, as
it does not increase the number of inner edges of $H$.	

\begin{lemma}\label{cl:internal-edges}
Graph $H$ has at most $2n+k-9$ inner edges, assuming that
all the $2n+k$ outer edges of $H$ are present.
\end{lemma}
\begin{proof}
Since $H$ is
  outerplanar, it follows that $H$ cannot have more than $2n+k-3$
  inner edges. Let $\mathcal{T}$ be the \emph{weak dual} of~$H$,
  i.e., the graph whose vertices are the internal faces of~$H$ and
  whose edges connect pairs of faces sharing an edge in~$H$.

  Since $H$ is biconnected outerplanar, graph $\mathcal{T}$ is a
  tree, and thus it has at least two leaves. Let $f$ be a leaf of
  $\mathcal{T}$ and let $e=(x,y)$ be the unique inner edge incident
  to $f$. Assume w.l.o.g.~that $x$ is a terminal of a fan-bundle of
  vertex $v_i$ of $G$, while $y$ is a terminal of a fan-bundle of a
  vertex of $G$ that follows vertex $v_i$ in the clockwise order of
  the vertices of $G$ along its outer face.
  
	\begin{figure}[tb]
		\centering
		\subfloat[\label{fig:threeFace-case1}{}]
		{\includegraphics[page=1,width=0.3\textwidth]{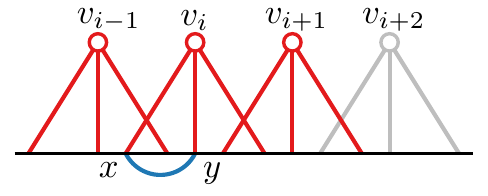}}
		\subfloat[\label{fig:threeFace-case2}{}]
		{\includegraphics[page=2,width=0.3\textwidth]{faces-claim2}}
		\subfloat[\label{fig:threeFace-case3}{}]
		{\includegraphics[page=3,width=0.3\textwidth]{faces-claim2}}
		
		\subfloat[\label{fig:threeFace-case4}{}]
		{\includegraphics[page=4,width=0.3\textwidth]{faces-claim2}}
		\subfloat[\label{fig:threeFace-case5}{}]
		{\includegraphics[page=5,width=0.3\textwidth]{faces-claim2}}
		\caption{
		The five possible cases in which $x$ and $y$ are at
        distance $2$ along the outer face of $H$.}
		\label{fig:threeFace}
	\end{figure}

  The first observation is that $x$ and $y$ cannot be at distance $2$ along the outer face of $H$. 
  For a proof by contradiction, assume that $x$ and $y$ are at distance $2$ along the outer face of $H$. 
  We distinguish the following three cases:
  
  \begin{enumerate}[label={C.}\arabic*]
    \item \label{c:d2.1} Vertex $x$ is the terminal of the left fan-bundle of vertex $v_i$ of $G$. 
    In this case, we distinguish two subcases; vertex $v_i$ has a middle fan-bundle or not 
    (see Figs.~\ref{fig:threeFace-case1} and~\ref{fig:threeFace-case2}, respectively).
    In the former case, edge $(x,y)$ represents a self-loop in $G$, 
    while in the latter case, edge $(x,y)$ is already represented by an outer edge of $H$ (in particular, by the outer edge of $H$ that corresponds to the connection of the right fan-bundle of $v_i$ with the left fan-bundle of $v_{i+1}$).
    
	\item \label{c:d2.2} Vertex $x$ is the terminal of the middle fan-bundle of vertex $v_i$ of $G$.
	Since $x$ and $y$ are at distance $2$, it follows that $y$ is the terminal of the right fan-bundle of vertex $v_i$ of $G$;
	see Fig.~\ref{fig:threeFace-case3}.
	This directly implies that $(x,y)$ represents a self-loop in $G$.
	
	\item \label{c:d2.3}  Vertex $x$ is the terminal of the right fan-bundle of vertex $v_i$ of $G$.
	In this case, we distinguish two subcases; vertex $v_{i+1}$ has a middle fan-bundle or not 
    (see Figs.~\ref{fig:threeFace-case4} and~\ref{fig:threeFace-case5}, respectively).
    In the former case, the middle fan-bundle of $v_{i+1}$ does not contain an edge of $G$ and therefore is redundant, 
    while in the latter case, edge $(x,y)$ is already represented by an outer edge of $H$. 
  \end{enumerate}

  From the above case analysis, it follows that $x$ and $y$ cannot be at distance $2$ along the outer
  face of $H$, which implies that $f$ cannot be a triangular face.
  
  Consider now the cases in which $x$ and $y$ are at distance $3$
  or $4$ along the outer face of $H$; see Fig.~\ref{fig:fdeg34}.
  With a case analysis similar to the one above, we prove that the
  case in which  edge $e$ connects
  the right fan-bundle of $v_i$ with the left fan-bundle of
  $v_{i+3}$, is the only case in which edge $e$%
  \begin{inparaenum}[(i)]
  \item does not represent a self-loop of $G$, or 
  \item does not represent an edge of $G$ that is already represented by any outer edge of $H$, or
  \item does not yield a redundant middle fan-bundle; see Fig.~\ref{fig:fdeg3-7}.
  \end{inparaenum}
  Note that in this case the distance between $x$ and $y$ is
  indeed $3$, and none of $v_{i+1}$ and $v_{i+2}$ has a middle
  fan-bundle.

  	\begin{figure}[p]
		\centering
		\subfloat[\label{fig:fdeg3-1}{Edge $(x,y)$ is represented as outer edge.}]
		{\includegraphics[page=1,width=0.31\textwidth]{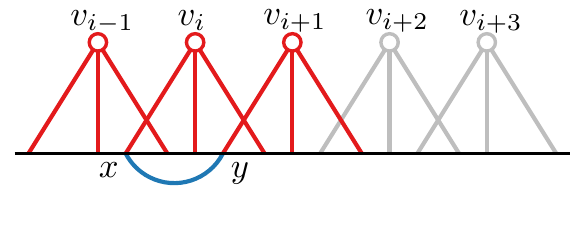}}
		\hfil
		\subfloat[\label{fig:fdeg3-2}{Edge $(x,y)$ represents a self-loop of $G$.}]
		{\includegraphics[page=2,width=0.31\textwidth]{faces-degree3}}
		\hfil
		\subfloat[\label{fig:fdeg3-3}{Edge $(x,y)$ is represented as outer edge.}]
		{\includegraphics[page=3,width=0.31\textwidth]{faces-degree3}}
		\hfil
		\subfloat[\label{fig:fdeg3-4}{Edge $(x,y)$ is represented as outer edge.}]
		{\includegraphics[page=4,width=0.31\textwidth]{faces-degree3}}
		\hfil
		\subfloat[\label{fig:fdeg3-5}{Edge $(x,y)$ is represented as outer edge.}]
		{\includegraphics[page=5,width=0.31\textwidth]{faces-degree3}}
		\hfil
		\subfloat[\label{fig:fdeg3-6}{Edge $(x,y)$ is represented as outer edge.}]
		{\includegraphics[page=6,width=0.31\textwidth]{faces-degree3}}
		\hfil
		\subfloat[\label{fig:fdeg3-7}{The only possible case.}]
		{\includegraphics[page=7,width=0.31\textwidth]{faces-degree3}}
		\hfil
		\subfloat[\label{fig:fdeg4-1}{Edge $(x,y)$ represents a self-loop of $G$.}]
		{\includegraphics[page=1,width=0.31\textwidth]{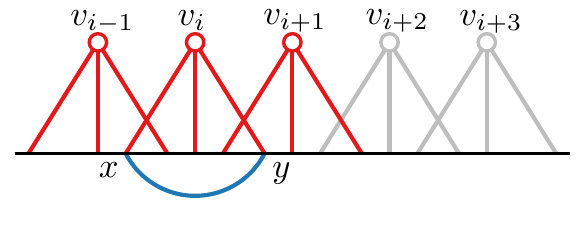}}
		\hfil
		\subfloat[\label{fig:fdeg4-2}{Edge $(x,y)$ is represented as outer edge.}]
		{\includegraphics[page=2,width=0.31\textwidth]{faces-degree4}}
		\hfil
		\subfloat[\label{fig:fdeg4-3}{Edge $(x,y)$ is represented as outer edge.}]
		{\includegraphics[page=3,width=0.31\textwidth]{faces-degree4}}
		\hfil
		\subfloat[\label{fig:fdeg4-4}{The middle fan-bundle of $v_{i+1}$ is redundant.}]
		{\includegraphics[page=4,width=0.31\textwidth]{faces-degree4}}
		\hfil
		\subfloat[\label{fig:fdeg4-5}{Edge $(x,y)$ is represented as outer edge.}]
		{\includegraphics[page=5,width=0.31\textwidth]{faces-degree4}}
		\hfil
		\subfloat[\label{fig:fdeg4-6}{The middle fan-bundle of $v_{i+1}$ is redundant.}]
		{\includegraphics[page=6,width=0.31\textwidth]{faces-degree4}}
		\hfil
		\subfloat[\label{fig:fdeg4-7}{The middle fan-bundle of $v_{i+2}$ is redundant.}]
		{\includegraphics[page=7,width=0.31\textwidth]{faces-degree4}}
		\hfil
		\subfloat[\label{fig:fdeg4-8}{Edge $(x,y)$ is represented as outer edge.}]
		{\includegraphics[page=8,width=0.31\textwidth]{faces-degree4}}		
		\caption{
		All different configurations in which $x$ and $y$ are: 
		(a)-(g) at distance $3$, and
		(h)-(o) at distance $4$ along the outer face of $H$.
		The different subcases arise based on whether 
		$e=(x,y)$ starts from the left, middle or right fan-bundle of vertex $v_i$ of $G$, and 
		on whether the vertices $v_i,\ldots,v_{i+3}$ have a middle fan-bundle or not.
		Note that the case illustrated in (g) is possible.
		For the remaining, the contradiction is given at the corresponding caption of each case.}
		\label{fig:fdeg34}
	\end{figure}
   
  	\begin{figure}[p]
		\centering
		\subfloat[\label{fig:fdeg5-1}{The middle fan-bundle of $v_i$ is redundant.}]
		{\includegraphics[page=1,width=0.31\textwidth]{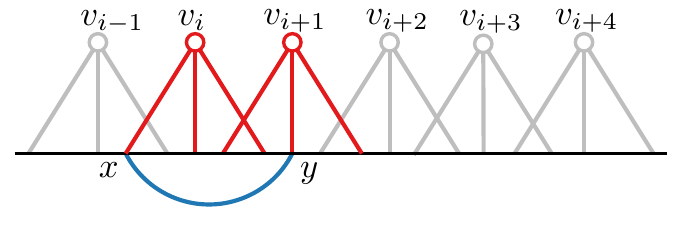}}
		\hfil
		\subfloat[\label{fig:fdeg5-2}{The middle fan-bundle of $v_{i+1}$ is redundant.}]
		{\includegraphics[page=2,width=0.31\textwidth]{faces-degree5}}
		\hfil
		\subfloat[\label{fig:fdeg5-3}{The middle fan-bundle of $v_i$ is redundant.}]
		{\includegraphics[page=3,width=0.31\textwidth]{faces-degree5}}
		\hfil
		\subfloat[\label{fig:fdeg5-4}{Edge $(x,y)$ is represented as outer edge.}]
		{\includegraphics[page=4,width=0.31\textwidth]{faces-degree5}}
		\hfil
		\subfloat[\label{fig:fdeg5-5}{The middle fan-bundle of $v_{i+1}$ is redundant.}]
		{\includegraphics[page=5,width=0.31\textwidth]{faces-degree5}}
		\hfil
		\subfloat[\label{fig:fdeg5-6}{The degree of $v_{i+1}$ is less than $5$.}]
		{\includegraphics[page=6,width=0.31\textwidth]{faces-degree5}}
		\hfil
		\subfloat[\label{fig:fdeg5-7}{The degree of $v_{i+1}$ is less than $5$.}]
		{\includegraphics[page=7,width=0.31\textwidth]{faces-degree5}}
		\hfil
		\subfloat[\label{fig:fdeg5-8}{The middle fan-bundle of $v_{i+1}$ is redundant.}]
		{\includegraphics[page=8,width=0.31\textwidth]{faces-degree5}}
		\hfil
		\subfloat[\label{fig:fdeg5-9}{The middle fan-bundle of $v_{i+2}$ is redundant.}]
		{\includegraphics[page=9,width=0.31\textwidth]{faces-degree5}}
		\hfil
		\subfloat[\label{fig:fdeg5-10}{The middle fan-bundle of $v_{i+1}$ is redundant.}]
		{\includegraphics[page=10,width=0.31\textwidth]{faces-degree5}}
		\hfil
		\subfloat[\label{fig:fdeg5-11}{The degree of $v_{i+2}$ is less than $5$.}]
		{\includegraphics[page=11,width=0.31\textwidth]{faces-degree5}}
		\hfil
		\subfloat[\label{fig:fdeg5-12}{The degree of $v_{i+2}$ is less than $5$.}]
		{\includegraphics[page=12,width=0.31\textwidth]{faces-degree5}}
		\caption{
		All different configurations in which $x$ and $y$ are at distance $5$ along the outer face of $H$. 
		The different subcases arise based on whether 
		$e=(x,y)$ starts from the left, middle or right fan-bundle of vertex $v_i$ of $G$, and 
		on whether the vertices $v_i,\ldots,v_{i+3}$ have a middle fan-bundle or not.
		The contradiction is given at the corresponding caption of each case.}
		\label{fig:fdeg5}
	\end{figure}
  
  	\begin{figure}[tb]
		\centering
		\subfloat[\label{fig:fourFace-t1}{}]
		{\includegraphics[page=11,width=0.33\textwidth]{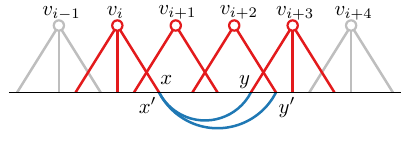}}
		\subfloat[\label{fig:fourFace-t2}{}]
		{\includegraphics[page=10,width=0.33\textwidth]{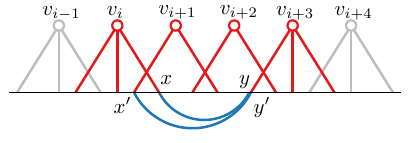}}
		\subfloat[\label{fig:fourFace-q3}{}]
		{\includegraphics[page=7,width=0.33\textwidth]{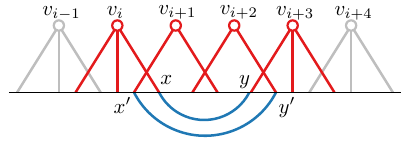}}
		
		\subfloat[\label{fig:fourFace-q1}{}]
		{\includegraphics[page=8,width=0.33\textwidth]{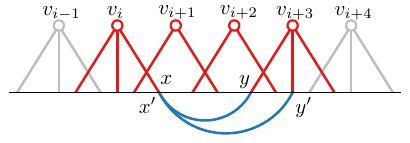}}					
		\subfloat[\label{fig:fourFace-q2}{}]
		{\includegraphics[page=9,width=0.33\textwidth]{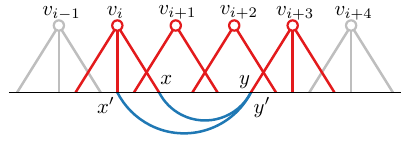}}

		\subfloat[\label{fig:fourFace-q4}{}]
		{\includegraphics[page=14,width=0.33\textwidth]{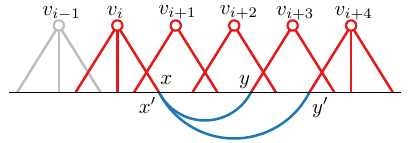}}		
		\subfloat[\label{fig:fourFace-q5}{}]
		{\includegraphics[page=13,width=0.33\textwidth]{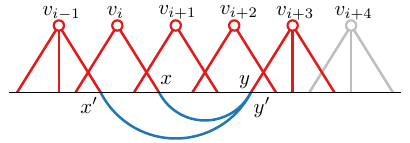}}		
		\caption{
		All different configurations in which $x'$ and $y'$ are: 
		(a)-(b) at distance $4$, and
		(c)-(g) at distance $5$ along the outer face of $H$.}
	\end{figure}

  We now claim that $x$ and $y$ cannot be at distance $5$ (or
  more). In Fig.~\ref{fig:fdeg5}, we illustrate all cases that can
  appear when the distance between $x$ and $y$ is exactly $5$. In all cases but one,
  either a middle fan-bundle is redundant or a vertex of $G$ has
  degree less than~$5$, both of which form a contradiction (note that
  this also holds when $x$ and $y$ are at distance greater than
  $5$). The exceptional case is illustrated in
  Fig.~\ref{fig:fdeg5-4}. In this particular case, if $x$ and $y$
  are at distance $5$, then the contradiction is based on the fact
  that the edge $(x,y)$ is represented as outer edge in $H$.
  On the other hand, if $x$ and $y$ are at distance greater than
  $5$, then $v_{i+1}$ has degree less than~$5$, which is again a
  contradiction. Hence, our claim follows. We can
  therefore assume that $x$ and $y$ are at distance $3$, as in
  Fig.~\ref{fig:fdeg34}.
  
  Let $f'$ be the face that is incident to $e$ and different from $f$.
  We prove that $f'$ cannot be of degree~$2$ in $\mathcal{T}$.
  Suppose for a contradiction that $f'$ has degree~$2$ in $\mathcal{T}$. 
  Hence, there is only one inner edge $e'=(x',y')$ of $H$ incident to $f'$ that is different from $e$.
  Since $x$ and $y$ are at distance $3$ along the outer face of $H$, it follows that 
  $x'$ and $y'$ are at distance at least $4$ along the outer face of $H$.  
  By simplicity, it follows that $x'$ and $y'$ cannot be at distance $4$ along the outer face of~$H$; see
  Figs.~\ref{fig:fourFace-t1} and~\ref{fig:fourFace-t2}.
  Thus, face $f'$ has at least four vertices or equivalently $x'$
  and $y'$ are at distance at least $5$ along the outer face
  of~$H$. Our next claim is that $x'$ and $y'$ cannot be at
  distance $5$. To prove this claim, we distinguish two cases, based
  on whether one of $x'$ and $y'$ coincides with one of~$x$
  and~$y$, or not. In the first case, either $x'=x$ or
  $y'=y$ holds; see Figs.~\ref{fig:fourFace-q1}-\ref{fig:fourFace-q5}. Hence, either the degree of $v_{i+2}$ or
  the degree of $v_{i+1}$ is less than $5$, which contradicts our
  assumption that every vertex has degree at least $5$. 
  In the second case, illustrated in Fig.~\ref{fig:fourFace-q3}, edge $e'$
  represents edge $(v_{i+1},v_{i+2})$ of $G$, which is already
  represented by an outer edge of $H$, and our claim follows.
  Finally, it is not difficult to observe that if $x'$ and $y'$ are
  at distance greater than $5$ along the outer face of $H$, then
  either the degree of $v_{i+2}$ or the degree of $v_{i+1}$ is less
  than $5$, which contradicts our assumption that every vertex has
  degree at least $5$. This concludes the proof that $f'$ cannot be of degree~$2$ in $\mathcal{T}$.

  Let $f_1,\dots,f_\ell$ be the leaves of $\mathcal{T}$. 
  Let also $f_1',\ldots,f_\ell'$ be the parents of $f_1,\dots,f_\ell$ in $\mathcal{T}$, respectively. 
  Since $f_1',\ldots,f_\ell'$ cannot be of degree~$2$ in $\mathcal{T}$, it follows that $\ell \geq 3$.
  Furthermore, if $\ell = 3$, then $\mathcal{T}$ must be a star with three leaves. 
  This implies that $H$ has exactly three inner edges, which is less than $2n+k-9$, since $n \geq 7$.
  Hence, we may assume w.l.o.g.~that $\ell \geq 4$ and therefore $H$ has at most $(2n+k-3)-4 = 2n+k-7$ inner edges.
  Note that if $\ell \geq 6$, then $H$ has at most $(2n+k-3)-6 = 2n+k-9$ inner edges. 
  Hence, we only have to consider the case where $4 \leq \ell \leq 5$.
  We distinguish two cases: $\mathcal{T}$ is a star or not.
  In the first case, $H$ has either four or five inner edges, which is less than $2n+k-9$, since $n \geq 7$.
  In the second case, there exist two faces in $\{f_1',\ldots,f_\ell'\}$ that are different from each other, say $f_1'$ and $f_\ell'$.
  We prove that $f_1'$ cannot be a triangular face in $H$; the proof for $f_\ell'$ is analogous. 
  Note that this proof also completes the proof of our claim, as it directly implies that $H$ cannot have more than 
  $(2n+k-3)-4-2 = 2n+k-9$ inner edges.
  Recall that $f_1'$ has at least two children in $\mathcal{T}$.
  Assume w.l.o.g.~that $f_1$ and $f_2$ are children of $f_1'$ and 
  let $e_1$ and $e_2$ be the unique inner edges incident to $f_1$ and $f_2$, respectively. 
  Since each of $e_1$ and $e_2$ must connect the right fan-bundle of a vertex $v_i$ 
  with the left fan-bundle of vertex $v_{i+3}$ as illustrated in Fig.~\ref{fig:fdeg3-7}, 
  it follows that $e_1$ and $e_2$ cannot share an endpoint. 
  Hence, $f_1'$ is not triangular. 
  This concludes the proof of the lemma.   
\end{proof}

We are now ready to prove the theorem about the edge density of $2$-sided
outer-\fb graphs.

\begin{theorem}\label{thm:2sided-outer-density}
  A $2$-sided outer-\fb graph $G$ with $n$ vertices has at most $4n-9$
  edges, where $n \geq 3$. This is a tight bound for all $n\geq 6$.
\end{theorem}
\begin{proof}
  Our proof is by induction on $n$.
  For the base case, observe that all graphs with $n \leq 6$ vertices have at
  most $n(n-1)/2 \leq 4n-9$ edges. Since the complete graph $K_6$ is $2$-sided
  outer-\fb (see Fig.~\ref{fig:k6}), it follows that all graphs with $n \leq 6$ vertices
  are in fact $2$-sided outer-\fb.

  For the inductive step, assume that $G$ has $n \geq 7$ vertices and
  let $\Gamma$ be a $2$-sided outer-\fb drawing of $G$. 
  Let also $v_{1},\ldots,v_{n}$ be the vertices of $G$ as they appear in
  clockwise order along the outer face of $\Gamma$.
  By the induction hypothesis, all $2$-sided outer-\fb graphs
  with $n' < n$ vertices have at most $4n'-9$ edges. We show that $G$ has at
  most $4n-9$ edges, as well.

  We first consider the case in which there exists a vertex $v_i$, for some $i=1,\ldots,n$, with degree
  at most $4$ in $G$. Since by the induction hypothesis the graph obtained by removing $v_i$ and all its
  incident edges, has at most $4(n-1)-9$ edges, we have that $G$
  has at most $4(n-1)-9+4= 4n - 9$ edges. Thus in the following we will assume
  that each vertex of $G$ has degree at least $5$.

	As discussed above, we assume w.l.o.g.\ that the right
  fan-bundle of $v_i$ crosses the left fan-bundle of $v_{i+1}$,
	and that the edge $(v_i,v_{i+1})$ is  represented with a crossing-free
	unbundled part connecting their terminals.
	
  We first consider the case in which there is a crossing between middle
  fan-bundles of two different vertices $v_i$ and $v_j$; see
  Fig.~\ref{fig:middlebundles}. We can assume that $v_i$ and $v_j$ are not
  consecutive along the outer face of $\Gamma$, as otherwise these crossing middle fan-bundles would isolate the
  other two crossing fan-bundles of~$v_i$ and $v_j$ (the right of~$v_i$ and the
  left of $v_j$, or vice versa), which could then be removed. Also, we can
  assume that the edge $(v_i,v_j)$ belongs to $G$, as otherwise we can add it 
  without violating the $2$-sided outer-\fby of $G$
  (see the dotted edge in Fig.~\ref{fig:middlebundles}). This implies that there
  is a second pair of crossing fan-bundles on the other side of $(v_i,v_j)$, as
  otherwise we can add them (although we possibly do not use them; 
  see the dashed fan-bundles in Fig.~\ref{fig:middlebundles}). Thus, the
  edge $(v_i,v_j)$ splits $\Gamma$ into two $2$-sided outer-\fb drawings
  $\Gamma_1$ and $\Gamma_2$ of two graphs $G_1$ and $G_2$, both containing
  vertices $v_i$ and $v_j$ and the edge between them. Let $n_1$ and $n_2$ be the
  number of vertices in~$G_1$ and $G_2$, respectively; note that $n_1+n_2=n+2$.
  By the induction hypothesis,~$\Gamma_1$ and $\Gamma_2$ have at most $4n_1-9$ and
  $4n_2-9$ edges, respectively. Since $(v_i,v_j)$ belongs to both $\Gamma_1$
  and~$\Gamma_2$, we have that $\Gamma$ has at most
  $(4n_1-9) + (4n_2-9) - 1 = 4(n+2)-19 = 4n-11 < 4n-9$~edges.

  To complete the proof, we consider the case in which no pair of middle fan-bundles
  cross in $\Gamma$. In this case, we can assume w.l.o.g.\ that each vertex is incident to
  at most one middle fan-bundle, as otherwise we could merge all its middle
  fan-bundles into one. Hence, Lemmas~\ref{cl:outer-face} and~\ref{cl:internal-edges}
	apply and we can conclude that $H$ has at most $4n-9$ edges.
This also implies that $G$ has at most $4n-9$ edges, since the only edges that could be drawn
in $\Gamma$ without using fan-bundles are between consecutive vertices $v_i$ and
$v_{i+1}$, but these edges are already in $H$.

The fact that the bound is tight follows from Lemma~\ref{lem:waterlily}.
Hence, the statement follows.
\end{proof}

In the following theorem, we study the edge density of $2$-sided $2$-layer \fb graphs. 
The upper bound is an immediate consequence of Theorem~\ref{thm:2sided-outer-density}. 
The corresponding lower bound is based on a construction similar to the one presented in the proof of Lemma~\ref{lem:waterlily}.

\begin{theorem}\label{thm:2sided-2layer-density}
  A $2$-sided $2$-layer \fb graph with $n \ge 3$
  vertices has at most $3n-7$ edges, while there exist $2$-sided $2$-layer \fb
  graphs with $n \geq 10$ vertices and $2n-4$~edges.
\end{theorem}
\begin{proof}
  Let $G$ be a $2$-sided $2$-layer \fb graph with $n\ge3$ vertices. As in the proof
  of Theorem~\ref{thm:density_layered}, we observe that one can add $n-2$ edges
  in $G$ and obtain a new graph $G'$ that is $2$-sided outer-\fb. Since by
  Theorem~\ref{thm:2sided-outer-density} graph $G'$ cannot have more than $4n-9$
  edges, it follows that $G$ cannot have more than $3n-7$ edges. For the
  corresponding lower bound, observe that all vertices of the graph of
  Fig.~\ref{fig:layeredlily} have degree exactly $4$, except for the vertices
  drawn as filled and non-filled squares, which have degrees~$2$ and $3$, respectively. 
  Hence, this graph has in total $2n-4$ edges. 
\end{proof}

For the general case, we have given in Lemma~\ref{lem:2sided-density} a lower bound on the edge density. 
In the following, we focus on a linear upper bound.

	Consider a $2$-sided \fb drawing $\Gamma$ of a maximally dense graph $G$,
  which contains the maximum number of uncrossed edges. Let $n$ and $m$ be the
  number of vertices and edges of $G$, respectively. To give an upper bound for
  $m$, we observe that each edge of $G$ can be identified by its unbundled part
  in $\Gamma$, which is unique for each edge.

	We proceed by defining a planar auxiliary subgraph $G_p$ of $G$, with $n_p$
  vertices and $m_p$ edges, as follows. Graph $G_p$ has the same vertex set as
  $G$, and so $n_p=n$, and contains all uncrossed edges of $G$ in $\Gamma$.
  Since $\Gamma$ contains a maximum number of uncrossed edges, it follows that
  for each pair of crossing fan-bundles $B_u$ and $B_v$, graph $G_p$ contains
  the base edge $(u,v)$ of $B_u$ and $B_v$ (note that the base edge of $B_u$
  and $B_v$ might occur several times in $G_p$, but such copies are pairwise
  non-homotopic). Hence, by the Euler's formula for planar graphs, it follows that $m_p \leq 3n-6$.

	Next, we create another planar graph $G_p'$, with $n_p'$ vertices and
  $m_p'$ edges, consisting of the vertices of $G$ and the terminals of the
  fan-bundles of $\Gamma$, which we call \emph{terminal vertices}. For each pair of
  crossing fan-bundles $B_u$ and $B_v$ with terminals $t_u$ and $t_v$, graph
  $G_p'$ contains edges $(u,t_v)$, $(t_v,t_u)$, $(t_u,v)$, and either edge $(u,t_u)$
  or edge $(v,t_u)$; see Fig.~\ref{fig:replacement}. We refer to these edges
  as \emph{bridging edges}, since they bridge vertices of the original graph
  with terminal vertices. Finally, for each unbundled part of each edge in
  $\Gamma$, graph $G_p'$ has an edge connecting the corresponding terminal
  vertices of $G_p'$. By construction, graph $G_p'$ is planar. If we denote by
  $t$ the number of terminal vertices of $G_p'$, then $n_p' = n+t$ and
  since~$G_p'$ is planar $m_p' \leq 3(n+t)-6$ holds.

\begin{figure}[t]
  \begin{minipage}[b]{.49\textwidth}
    \centering
    \includegraphics[page=5]{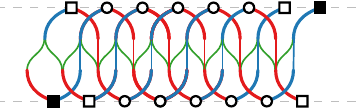}
    \caption{Illustration for the proof of Theorem~\ref{thm:2sided-2layer-density}.}
    \label{fig:layeredlily}
  \end{minipage}
  \hfill
  \begin{minipage}[b]{.46\textwidth}
    \centering
    \includegraphics[page=6]{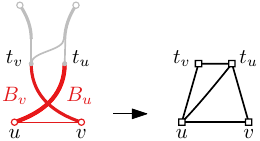} 
    \caption{Illustration for the proof of Theorem~\ref{thm:general_bound}.}
    \label{fig:replacement}
  \end{minipage}
\end{figure}

  Observe, however, that for each pair of terminal vertices $t_u$ and $t_v$
  corresponding to the terminals of two crossing fan-bundles anchored at two
  vertices $u$ and $v$, respectively, all the four bridging edges incident to
  $t_u$ and $t_v$ are not in correspondence with edges of the original graph $G$. 
  Hence, the number of edges that actually correspond to distinct edges of $G$ is equal to $m_p'$ minus the
  number of bridging edges, which is equal to $2t$ since every two terminal
  vertices determine four bridging edges. This implies that:

	\begin{equation}
	m \leq 3(n+t)-6-2t=3n+t-6
	\label{eq:mup}
	\end{equation}

	Note that the arguments presented so far would already give a linear upper bound
  on the number of edges of $G$; in fact, since we can associate at most four
  terminal vertices to each edge of $G_p$ (as each of these edges can have at
  most two crossing fan-bundles on each side), we have that $t \leq 4 m_p$,
  which gives $t \leq 4\cdot(3n-6)=12n-24$ and thus $m \leq 3n+t-6 \leq 15n-30$.

	In order to improve this bound, we will show in the following that the
  value of $m_p$ is actually significantly smaller than $3n-6$. The general
  idea is that, if $G_p$ contains a small face $f$ (which is always the case if
  $m_p$ is equal or close to $3n-6$), then it is not possible for all the edges
  incident to $f$ to have fan-bundles inside $f$ without having multiple edges
  in $G$; note that this reduces the number of terminal vertices in $G_p'$, and
  hence its number of edges. This is clear, for example, when~$f$ is
  triangular, and thus all the connections that could be represented by fan-bundles
  inside $f$ are already represented by the three edges incident to $f$;
  in this case, in fact, none of these three edges incident to $f$ may have
  fan-bundles inside it. We formalize this concept in the following.

	Consider any (possibly non-simple) $k$-cycle of $G_p$, with $2 \leq k \leq 6$,
  delimiting a face of some connected component of $G_p$ in $\Gamma$. If
  this $k$-cycle also delimits a face of $G_p$, then we call it \emph{empty};
  otherwise we call it \emph{non-empty}. Note that if $k=2$ then the cycle must
	be non-empty (as otherwise we would have a pair of homotopic parallel edges).
	Also note that a non-empty $k$-cycle contains in
  its interior all the vertices and edges of at least another connected
  component of $G_p$. Further, there exists a non-connected face of $G_p$ whose boundary
	consists of this non-empty $k$-cycle and the outer boundaries of all the
	components contained in it.
	We denote by $f_k$ the number of empty $k$-cycles
  and by $\phi_k$ the number of non-empty $k$-cycles in $\Gamma$.
	Hence, $\sum_{k=3}^\infty f_k$ and $\sum_{k=2}^\infty \phi_k$ are the numbers of connected
	and non-connected faces of $G_p$, respectively.
	
	For the following lemma, recall that an edge is accounted twice for a face
	if both its sides are incident to this face.
	
	\begin{lemma}\label{lem:smallFaces}
	For a non-empty $k$-cycle $C$ with $k=2,3,4$, the face $f$ of $G_p$ that is delimited
	by $C$ has at least 5 incident edges.
	\end{lemma}
	
	\begin{proof}
	In order to prove the statement for $k=3,4$, it is sufficient to show that
	at least one connected component of $G_p$ in the interior of $C$ is not an
	isolated vertex. By maximality, in $\Gamma$ there must be a crossing of two
	fan-bundles in the interior of $C$. Let these crossing fan-bundles be anchored
	at vertices $u$ and $v$. Recall that $G_p$ contains the edge $(u,v)$, since it is the
	base edge of this bundle crossing. This implies that either both $u$ and $v$ belong to
	$C$ or none of them belongs to $C$. In the latter case $u$ and $v$ belong to a
	connected component of $G_p$ in the interior of $C$ which is not an isolated vertex and the
	statement follows. So we may assume that for every fan-bundle crossing, the base edge
	is an edge of $C$. Consider the graph $H$ induced by the isolated vertices in the interior
	of $C$ and by the terminals of the fan-bundles whose base edges are edges of $C$.
	By our previous observation this graph is plane. Note that when $k=3$ there is no
	edge in $H$ between two terminals, as otherwise any such edge would represent
	an edge of $C$; a contradiction. On the other hand, when $k=4$, there can be at most
	two edges connecting terminal vertices in $H$ (corresponding to the two diagonals of $C$; 
	see Fig.~\ref{fig:2sided-cycles-1}).
	In both cases, we conclude that there is at least one isolated vertex in the interior of $C$
	that is incident to the outer face of $H$. Hence, this vertex can be connect to a vertex
	of $C$ with a planar edge; a contradiction.
	
	Consider now the case $k=2$. Let $u$ and $v$ be the two vertices belonging to $C$.
	Using the same argument as above, we can conclude that there exists at least a
	connected component of $G_p$ in the interior of $C$ that is not an isolated vertex.
	If this component has at least three edges incident to its outer face then the
	statement follows. The same holds if there exists more than one component that
	is not an isolated vertex. Hence we can conclude that there exists a single component $\sigma$
	that is not an isolated vertex, and that either $\sigma$ is an edge or its outer face is
	a pair of parallel edges. In both cases, $\sigma$ has only two vertices $w$ and $z$
	incident to its outer	face. If there is no isolated vertex in $f$ then $u,v,w$ and $z$ are the only
	vertices incident to $f$. This contradicts the fact that the pairs $\langle u,v\rangle$
	and $\langle w,z\rangle$ belong to different components of $G_p$, since at least
	one of the vertices $w$ or $z$ can be connected to $u$ or $v$ by a planar edge.
	So we can assume that there exists an isolated vertex incident to $f$.
	Since the isolated vertices are not incident to fan-bundles and since edges $(u,v)$ and
	$(w,z)$ are planar, the only possible connections between two terminals in $f$ are
	between a terminal of a fan-bundle anchored at $u$ or $v$ and a terminal of a fan-bundle
	anchored at $w$ or $z$. These connections split $f$ into at most 4
	regions (refer to the gray colored regions of Fig.~\ref{fig:2sided-cycles-2}). 
	Since each of these regions contains on its boundary at
	least one of the vertices $u,v,w$ or $z$, it is always possible to connect an
	isolated vertex to one of these four vertices by a planar edge; a contradiction.
	
	We conclude the proof by noting that for $k=5$ it is possible to have a non-empty $k$-cycle that contains
	in its interior only isolated vertices; for an illustration refer to Fig.~\ref{fig:2sided-cycles-3}.
	\end{proof}
	
	\begin{figure}[tb]
		\centering		
		\subfloat[\label{fig:2sided-cycles-1}{}]
		{\includegraphics[page=2]{2sided-cycles}}
		\hfil
		\subfloat[\label{fig:2sided-cycles-2}{}]
		{\includegraphics[page=3]{2sided-cycles}}
		\hfil
		\subfloat[\label{fig:2sided-cycles-3}{}]
		{\includegraphics[page=1]{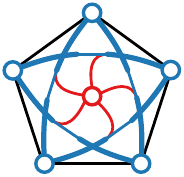}}
		\caption{Illustrations for the proof of Lemma~\ref{lem:smallFaces}.}
	\end{figure}

  In the following we will assume that every empty $k$-cycle with $3 \leq k \leq 6$
  has no terminal vertex in its interior. We observe that for $k\ne3$
  this results in an underestimation of the number of edges, which we will
  compensate in the final computation by considering each of the cases
  independently.

	This assumption implies that the number $t$ of terminal vertices may be smaller
  than $4m_p$, and in particular it can be expressed as
	\begin{equation}
    t \leq 4m_p - 6f_3 - 8f_4 - 10f_5 - 12 f_6
	\label{eq:dup}
	\end{equation}
	Further, we can also express $m_p$ as a function on the number of the $k$-
  cycles. In particular, by using the fact that $2 m_p$ equals the sum of the size
	of all faces of $G_p$, and by using Euler's formula for disconnected planar graphs
  $m_p = n + f_p -1 -c_p$, where $f_p$ denotes the number of faces of $G_p$ and
  $c_p$ denotes the number of its connected components, we get:
	
	\begin{align*}
	&\mathrel{\phantom{+}} 3 f_3 + 4 f_4 + 5 (f_5 + \phi_2 + \phi_3 + \phi_5) + 6 (f_6 + \phi_4 + \phi_6) \\
	&+ 7 (f_p - (\phi_2 + f_3 + \phi_3 + f_4 + \phi_4 + f_5 + \phi_5 + f_6 + \phi_6)) \\
	&\leq 2 m_p = 2n + 2f_p -2 -2c_p
	\end{align*}
	where we use the coefficient 5 for $\phi_2$ and $\phi_3$ and the coefficient 6 for $\phi_4$
	due to Lemma~\ref{lem:smallFaces}. This yields:
	\begin{equation}
	5f_p - 4 f_3 - 3 f_4 - 2 (f_5 + \phi_2 + \phi_3 + \phi_5) - (f_6 + \phi_4 + \phi_6) \leq 2n -2 -2c_p
	\label{eq:fp}
	\end{equation}
	Observe that $c_p \geq \phi_2 + \phi_3 + \phi_4 + \phi_5 + \phi_6$, since each connected
  component of $G_p$ can be used to identify at most one face as non-empty. Thus,
  replacing $c_p$ in Eq.~\ref{eq:fp} we obtain:
	
	\begin{align*}
	&\mathrel{\phantom{\leq}} 5f_p - 4 f_3 - 3 f_4 - 2 (f_5 + \phi_2 + \phi_3 + \phi_5) - (f_6 + \phi_4 + \phi_6) \\
	&\leq 2n -2 -2(\phi_2 + \phi_3 + \phi_4 + \phi_5 + \phi_6)
	\end{align*}
	which yields:
	\[f_p \leq \frac{1}{5} (2n -2 + 4 f_3 + 3 f_4 + 2 f_5 + f_6)\]
	Applying again Euler's formula $m_p \leq n + f_p -2$ (using that $c_p \geq 1$), we obtain:
	\[m_p \leq \frac{1}{5} (7n - 12 + 4 f_3 + 3 f_4 + 2 f_5 + f_6)\]
	By Eq.~\ref{eq:dup} we have:
	\[t \leq \frac{4}{5} (7n - 12 + 4 f_3 + 3 f_4 + 2 f_5 + f_6) - 6f_3 - 8f_4 - 10f_5 - 12 f_6,\]
	which implies:
	\[t \leq \frac{1}{5}(28n- 48 - 14 f_3 - 28 f_4- 42 f_5 - 56 f_6)\]
	Hence, by Eq.~\ref{eq:mup} we might provide a bound for $m$, which is unfortunately underestimated, as we observed above:
	\[m \leq 3n + t - 6 = \frac{1}{5}(43n- 78 - 14 f_3 - 28 f_4- 42 f_5 - 56 f_6)\]
To compensate the underestimation of the number of edges, we conclude our
discussion by studying how many crossing edges can be drawn in the interior
of an empty $k$-cycle, for $k=3,\dots,6$. Namely, empty $3$-cycles (that is,
triangular faces) cannot have any edge in their interior, as discussed above.
Empty $4$-cycles can have at most two edges, namely those connecting
vertices at distance~$2$ along the $4$-cycle. For the number of edges of empty $k$-cycles
with $k=5,6$, we use as an upper bound the number of edges in the complete graph on $k$ vertices
minus $k$. We
thus have five edges for $k=5$ and ten edges for $k=6$. Hence, the final bound for
the number of edges of $G$ is:
	\[m \leq 3n + t - 6 + 2f_4 + 5 f_5 + 10 f_6 = \frac{1}{5}(43n- 78 - 14 f_3 - 18 f_4- 17 f_5 - 6 f_6)\]
	Hence, $G$ cannot have more than $(43n- 78)/5$ edges. Combining with the lower
	bound we proved in Lemma~\ref{lem:waterlily}, we obtain the following theorem.
	
\begin{theorem}\label{thm:general_bound}
  A $2$-sided \fb graph with $n \ge 3$ 
  vertices has at most $(43n - 78)/5$ edges, while there exist $2$-sided \fb
  graphs with $n \geq 9$ vertices and $6n-18$~edges.
\end{theorem}


\section{NP-completeness}
\label{sec:npcompleteness}

In this section, we prove that the problem of testing whether a graph $G$ with a
given rotation system~$R$ admits a $1$-sided or a $2$-sided \fb drawing
preserving~$R$ is NP-complete. We present the reduction for the $1$-sided model in detail, and we only highlight the differences for the $2$-sided model.

\begin{theorem}\label{thm:completeness}
Given a graph $G$ and a fixed rotation system $R$ of $G$, it is NP-complete to decide whether $G$ admits a $1$-sided \fb drawing preserving $R$.
\end{theorem}
\begin{proof}
  Membership in NP can be proved as for fan-planarity~\cite{BekosCGHK14}, which is in turn inspired by the corresponding proof for the crossing number~\cite{gj-cnnpc-83}.
  
We prove the NP-hardness by means of a reduction from problem \threepartition. The idea is
  based on a general scheme proposed by Bekos et al.~\cite{BekosCGHK14} to prove the
  NP-completeness of the fan-planarity problem with a fixed rotation system.
  Recall that an instance $\langle A,B \rangle$ of \threepartition consists of an integer~$B$
  and of a set $A = \{a_1, a_2, \ldots, a_{3m} \}$ of $3m$ integers
  such that $a_i \in (\frac{B}{4}, \frac{B}{2})$, for $i=1,2,\ldots, 3m$,
  and $\sum_{i=1}^{3m} a_i=mB$. Problem \threepartition asks
  whether $A$ can be partitioned into $m$ subsets $A_1, A_2, \ldots, A_m$, each
  of cardinality $3$, such that the sum of the numbers in each subset is
  exactly~$B$. Note that \threepartition is strongly NP-hard~\cite{gj-cigtnpc-79}.
  So, we may assume w.l.o.g.~that $B$ is bounded by a polynomial in~$m$.

	Given an instance $\langle A,B \rangle$ of \threepartition, we show how to
  construct in polynomial time an instance $\langle G, R \rangle$ of our
  problem such that there is a solution for $\langle A,B \rangle$ if and only if $G$ admits
  a $1$-sided \fb drawing preserving $R$.

	Central in our transformation is the so-called \emph{barrier gadget}. To describe this gadget, we first introduce a graph $H$ composed of seven vertices
  $a,b,c,d,e,f,g$; refer to Fig.~\ref{fig:2k23}. Graph $H$ contains cycle $(a,b,c,d,e,f)$,
  which is called \emph{boundary cycle} and whose edges are the \emph{boundary edges} of $H$, and two edges $(c,g)$ and $(f,g)$.  Also, for each vertex $u \in \{c,f,g\}$ and for each vertex
  $v \in \{a,b,d,e\}$, graph $H$ contains edge $(u,v)$. The rotation system of $H$ is such that
  the boundary cycle delimits its outer face in any drawing respecting this
  rotation system, while all the other edges (which are called \emph{inner}
  edges) are routed in its interior, as in Fig.~\ref{fig:2k23}. We refer to
  vertices $a$, $e$, and~$f$ as \emph{left-sided} and to $b$, $c$, and $d$ as
  \emph{right-sided}.

	To construct an $n$-vertex barrier gadget with $n \geq 7$, we employ
  $\lfloor (n-3)/4 \rfloor$ copies of the graph $H$, which we glue with each
  other by identifying the left-sided vertices of one copy with the right-sided
  vertices of the next copy; see Fig.~\ref{fig:barrier-gadget}. We fix the rotation system of the barrier gadget so that for each vertex, the edges belonging to the same copy of $H$ are consecutive around it.   
  We will use the
  barrier gadget in order to constrain the routes of some specific paths of~$G$.

	\begin{figure}[t]
		\centering
		\subfloat[\label{fig:2k23}{}]
		{\includegraphics[page=3]{3P-reduction}}
		\hfil
		\subfloat[\label{fig:barrier-gadget}{}]
		{\includegraphics[page=2]{3P-reduction}}

		\subfloat[\label{fig:reduction_outline}{}]
		{\includegraphics[scale=.8,page=1]{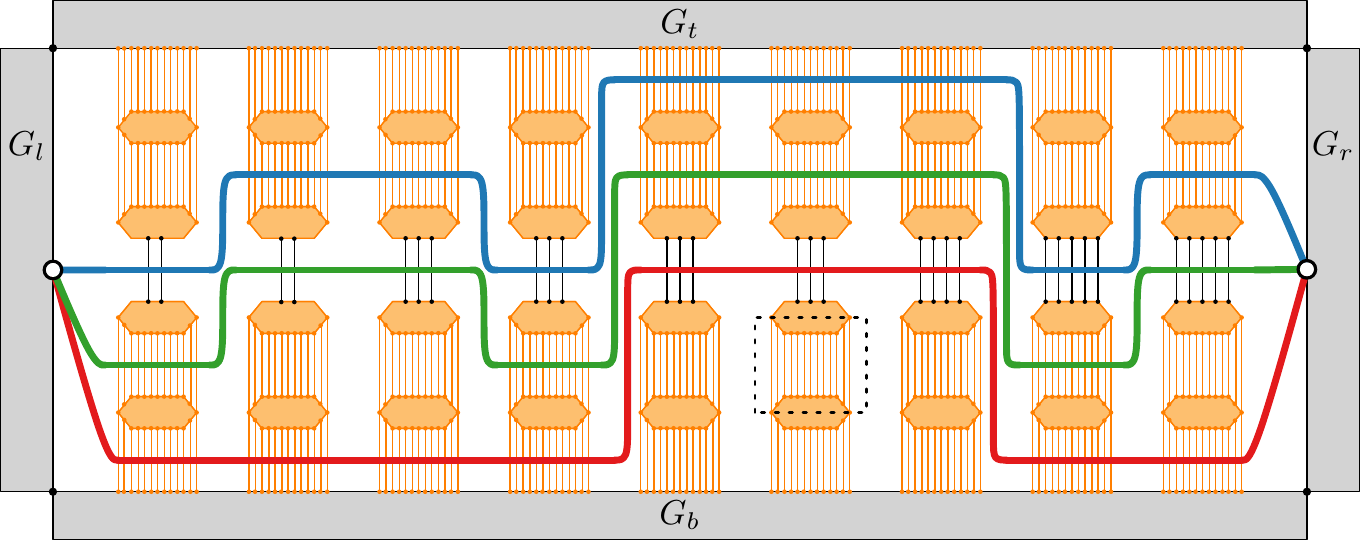}}
		\hfill
		\caption{(a)~The graph $H$ used in the construction of the barrier gadget,
      which is illustrated in~(b). In~(c) we provide the whole scheme of the
      reduction from \threepartition, for the case in which $m=3$,
      $A=\{2,2,2,3,3,3,4,5,6\}$, and $B = 10$. The transversal paths are routed
      according to the following solution of \threepartition:
      $A_1=\{2,3,5\}$, $A_2=\{2,3,5\}$, and $A_3=\{3,3,4\}$.
			The hexagonal regions are the obstacles; the dotted rectangle contains all
			the vertical edges of a (bottom) cell.}
		\label{fig:reduction}
	\end{figure}

	Consider now a biconnected $1$-sided \fb graph $G$ with rotation system~$R$ that contains as a
  subgraph a barrier gadget $G_b$. Let $\Gamma$ be any $1$-sided \fb drawing
  of~$G$ respecting~$R$. Observe that, by the choice of the rotation system, the boundary
  edges of~$G_b$ do not cross any other edge of~$G_b$, while all the inner
  edges have at least one crossing with another inner edge, except possibly for
  those incident to $g$. In particular, the inner edges incident to $a$ must
  share a fan-bundle anchored at~$a$, and those incident to $b$ must share a
  fan-bundle anchored at $b$, and these two fan-bundles must cross; analogously,
  two fan-bundles anchored at $d$ and $e$ must cross. This implies that no
  path~$\pi$ of~$G \setminus G_b$ can enter inside the boundary cycle of $G_b$ and
  cross an inner edge of $G_b$ in~$\Gamma$. On the other hand, if path $\pi$
  enters inside the boundary cycle of~$G_b$ without crossing any inner edge, then it must cross
  the same boundary edge a second time to exit this cycle (due to the biconnectivity of $G$).
  In other words, if a path $\pi$ enters $G_b$ in $\Gamma$, then it must exit it by
  using the same boundary edge, which is equivalent to not entering it at all.

	We construct an instance $\langle G, R \rangle$ of our problem based on an
  instance $\langle A,B \rangle$ of \threepartition as follows. We start our
  construction with the \emph{wall gadget}, which consists of a cyclic chain of
  four barrier gadgets $G_t$, $G_r$, $G_b$, and $G_\ell$ that surrounds the whole
  construction; see Fig.~\ref{fig:reduction_outline}. The barrier gadgets $G_t$
  and $G_b$ are called \emph{top} and \emph{bottom beams}, respectively, and
  contain exactly $4 \cdot (3mK+1) +3$ vertices each, where~$K$ is a large
  integer number, e.g., $K = B^2$. The barrier gadgets $G_\ell$ and $G_r$ are called
  \emph{left} and \emph{right} walls, respectively, and have only~11
  vertices each. In other words,~$G_t$ and~$G_b$ contain $3mK+1$ copies of $H$,
  while $G_\ell$ and $G_r$ contain only two copies of $H$. By the choice of the
  rotation system $R$ and of the vertices shared by two consecutive barrier
  gadgets, we may assume that $3mK$ vertices of each of $G_t$ and $G_b$, and
  one vertex of each of~$G_\ell$ and~$G_r$, are incident to the \emph{interior of
  the wall}, that is, the closed region delimited by the wall~gadget.

	The top and bottom beams are ``bridged'' to each other by a set of $3m$
  \emph{columns}; see Fig.~\ref{fig:reduction_outline} for an illustration of
  the case $m = 3$. Each \emph{column} contains $2m-1$ \emph{cells}, where a
  cell consists of a set of pairwise disjoint edges, called \emph{vertical edges}
  of that cell (see, e.g., the edges that are contained in the dotted rectangle in
	Fig.~\ref{fig:reduction_outline}).
	There are $m-1$ \emph{top cells}, one \emph{central cell},
  and~$m-1$ \emph{bottom cells}. Cells of the same column are separated by~$2m-2$
  barrier gadgets, called \emph{obstacles}, which have $4\cdot (K-1)+3$
  vertices each (see the hexagonal regions in Fig.~\ref{fig:reduction_outline}).
	The number of vertical edges of each of the $3m$ central cells
  depends on the elements of instance $A$. In particular, for
  $i =1,2, \ldots, 3m$, the central cell $C_i$ of the $i$-th column has exactly $a_i$ vertical
  edges connecting its delimiting obstacles. Each of the remaining cells has~$K$
  vertical edges. Hence, each of the top and bottom cells contains
  significantly more vertical edges than any central cell. We say that
  central cells are \emph{sparse}, while the top and the bottom cells are~\emph{dense}.

	The left and the right walls are ``bridged'' to each other by a set of $m$
  pairwise internally disjoint paths $\pi_1,\pi_2,\ldots,\pi_m$, called
  \emph{transversal paths}, which all originate from the same vertex of the left
  wall, called \emph{origin}, and terminate at the same vertex of the right
  wall, called \emph{destination}. Each of these paths has length $(3m-3)K+B$.

	Regarding the choice of the rotation system $R$, we define a cyclic order of
  the edges around each vertex that conforms with the following constraints.
	\begin{enumerate}[label=C.\arabic*:]
		\item all inner edges of each barrier gadget lie in the interior of its boundary cycle,
		\item the wall gadget is embedded such that $3mK+2$ vertices of each top
      and bottom beam and four vertices of each left and right wall are incident
      to the interior of the wall,
		\item all columns can be embedded in the interior of the wall without crossing each other,
		\item the vertical edges of each cell can be embedded without crossing each other, and
		\item the order of the edges of the transversal paths around the origin is
      the reverse of the corresponding order around the destination, which
      guarantees that the transversal paths can avoid crossing each other.
	\end{enumerate}
	
	This concludes our construction, which is clearly polynomial in $m$, 
	since we have assumed that $B$ is bounded by a polynomial in $m$.

	We now prove the equivalence, which is mainly based on the observation that
  each transversal path has to cross exactly $3$ sparse cells and exactly $3m-3$
  dense cells in any $1$-sided \fb drawing. This is due to the following fact.
  Since each transversal path has length $(3m-3)K+B$, it can cross at most $3m-3$
  dense cells in order to connect the origin to the destination. On the other
  hand, since no two different paths can cross the same cell in any $1$-sided
  \fb drawing, we have that if any transversal path crosses fewer than $3m-3$
  dense cells, then there must be another one that crosses more than $3m-3$ of
  these cells, and the claim follows.

	Suppose that the set $A$ admits a partition into subsets $A_1,A_2,\ldots,A_m$,
  each composed of three integers summing up to $B$. If one omits the
  transversal paths, then it is easy to compute a $1$-sided \fb drawing $\Gamma$
  of $G$ preserving $R$. It is essentially a drawing like the one depicted in
  Fig.~\ref{fig:reduction_outline}, where columns are next to each other
  in the interior of the wall. To complete the drawing, we embed the transversal
  paths $\pi_1,\pi_2,\ldots,\pi_m$ of $G$ in the partial drawing of $G$
  constructed so far under the following requirements:
	\begin{enumerate}[label=R.\arabic*]
		\item \label{r:1} transversal paths $\pi_1,\pi_2,\ldots,\pi_m$ do not cross each other,
		\item \label{r:2} transversal paths $\pi_1,\pi_2,\ldots,\pi_m$ do not cross any barrier gadget,
		\item \label{r:3} each cell is traversed by at most one transversal path (as otherwise $1$-sided \fby would be deviated), and
		\item \label{r:4} each transversal path passes through exactly $3$ sparse cells and $3m-3$ dense cells.
	\end{enumerate}

	We obtain a drawing satisfying these requirements as follows. For $j=1,2,\ldots,m$,
  let $A_j=\{a_\kappa, a_\lambda,a_\mu\}$, where
  $1 \leq \kappa, \lambda, \mu \leq 3m$. Then, in the drawing~$\Gamma$, the
  path $\pi_j$ will cross the $\kappa$-th, $\lambda$-th, and $\mu$-th vertical
  columns of $G$ through sparse cells, and
  the remaining vertical columns of $G$ through dense cells. Hence,
  Requirement~\ref*{r:4} is satisfied. The routing of the remaining transversal paths
  through the $\kappa$-th vertical column is done as follows. By construction,
  there exist $m-1$ cells above and $m-1$ cells below the sparse cell of the
  $\kappa$-th vertical column (all of which are dense). Hence, there exist at
  least as many available dense cells as transversal paths to route at each
  side of the sparse cell of the $\kappa$-th vertical column. Hence, we can
  route the remaining transversal paths through the $\kappa$-th vertical column
  such that Requirements \ref*{r:1}--\ref*{r:3} are also satisfied. The
  corresponding routings through the $\lambda$-th and $\mu$-th vertical columns
  of $G$ are symmetric. This implies that the drawing $\Gamma$ of $G$ is indeed
  $1$-sided \fb and preserves $R$.

	Suppose now that $G$ admits a $1$-sided \fb drawing $\Gamma$ preserving the
  rotation system~$R$. As already mentioned, each of the transversal paths
  crosses exactly~$3$ sparse cells and exactly $3m-3$ dense cells. In addition,
  $1$-sided \fby ensures that no two transversal paths pass through the same
  cell. With these two properties, we can construct a solution
  $A_1,A_2,\ldots,A_m$ of instance $\langle A, B \rangle$ of \threepartition as follows. Assume
  that path $\pi_j$ crosses the $\kappa$-th, $\lambda$-th, and $\mu$-th vertical
  columns of $\Gamma$ through sparse cells, where $1 \leq \kappa, \lambda, \mu \leq 3m$.
  Then, the $j$-th partition set $A_j$ of instance $\langle A, B \rangle$ of \threepartition will contain
  integers $\{a_\kappa,a_\lambda, a_\mu \}$. Since
  $a_\kappa + a_\lambda + a_\mu= B$, the solution constructed this way is indeed
  a solution of \threepartition for the instance~$\langle A,B \rangle$. This
  concludes our NP-hardness reduction.
\end{proof}

We observe that the NP-completeness of $2$-sided \fby with a given rotation
system can be proved as in Theorem~\ref{thm:completeness} with the following
modifications. Since each edge of the transversal path can be crossed twice
in the $2$-sided model, we double the number of vertical edges in the dense
and sparse cells. To avoid that two transversal paths cross the same cell, we
enforce that consecutive pairs of edges in the same cell cross; see
Fig.~\ref{fig:2sided-npcompleteness-regions}. For the barrier gadget, we use the graph of
Fig.~\ref{fig:main_2sided-npcompleteness-gadget}, which by the choice of the
rotation system cannot be crossed by any transversal path. We summarize these
observations in the following theorem.

\begin{figure}[t]
  \centering
  \subfloat[\label{fig:2sided-npcompleteness-regions}]{\includegraphics[page=2]{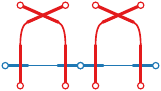}}
  \hfil
  \subfloat[\label{fig:main_2sided-npcompleteness-gadget}]{\includegraphics[page=6]{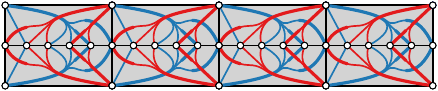}}
  \caption{
  (a) Edges in the~same cell in the $2$-sided model, and
  (b) the barrier gadget in the $2$-sided models.}
\end{figure}

\begin{theorem}\label{thm:2sided-completeness}
Given a graph $G$ and a fixed rotation system $R$ of $G$, it is NP-complete to determine whether $G$ admits a $2$-sided \fb drawing preserving $R$.
\end{theorem}

\section{Recognition and drawing algorithms}
\label{sec:recognition}

In this section, we present recognition and drawing algorithms for biconnected $1$-sided
$2$-layer \fb graphs, maximal $1$-sided $2$-layer \fb graphs, and triconnected
$1$-sided outer-\fb graphs. We also give a complete characterization of general $1$-sided $2$-layer
\fb graphs.

\subsection{$1$-sided $2$-layer \fblong graphs.}
\label{subsec:recog:2layer}

In this subsection, we present linear-time recognition and drawing algorithms for 
biconnected $1$-sided $2$-layer \fb graphs and maximal $1$-sided $2$-layer \fb graphs.
Since a $1$-sided $2$-layer \fb graph is by definition $2$-layer fan-planar, 
naturally our results build upon known results by 
Binucci et al.~\cite{DBLP:conf/gd/BinucciCDGKKMT15} for $2$-layer fan-planar graphs, 
who showed that a biconnected bipartite graph is maximal $2$-layer
fan-planar if and only if it is a \emph{snake}, i.e., a chain of
graphs~$G_1,\ldots,G_k$ such that each~$G_i$ is a complete bipartite
graph~$K_{2,h_i},h_i\ge 2$ that shares a pair of vertices, called
\emph{merged vertices}, with~$G_{i+1}$, and no vertex is shared by more
than two graphs; see Fig.~\ref{fig:2layer-snakes} for an illustration. Furthermore, they also showed that a biconnected bipartite graph
is $2$-layer fan-planar if and only if it is a spanning subgraph of a snake. 
Hence, every biconnected $2$-layer \fb graph
has to be a spanning subgraph of a snake. However, not every snake is $1$-sided
$2$-layer \fb, as we demonstrate in the following lemma.

\begin{figure}
  \centering
  \subfloat[\label{fig:2layer-snake}]{\includegraphics[page=1]{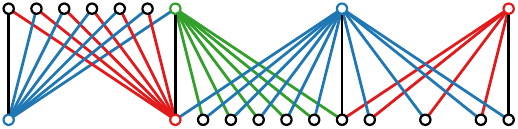}}
  \hfil
  \subfloat[\label{fig:2layer-babysnake}]{\includegraphics[page=2]{2layer-snakes}}
  \caption{Illustration of (a)~a snake and (b)~a baby snake.}
  \label{fig:2layer-snakes}
\end{figure}

\begin{lemma}\label{lem:$2$-layer-k24}
The complete bipartite graph $K_{2,3}$ is $1$-sided $2$-layer \fblong, 
while the complete bipartite graph $K_{2,4}$ is not $1$-sided $2$-layer \fblong.
\end{lemma}
\begin{proof}
  Let $\{a_1,a_2\}$ and $\{b_1,b_2,b_3\}$ be the two partition sets of $K_{2,3}$.
  Topologically, there is exactly one 
  $1$-sided $2$-layer \fb drawing of $K_{2,3}$ such that $x(a_1)<x(a_2)$ and
  $x(b_1)<x(b_2)<x(b_3)$, which is illustrated in Fig.~\ref{fig:2layer-k23}.
  The reason is that any $2$-layer drawing of $K_{2,3}$ is not crossing-free, 
  which implies that $a_1$ and $a_2$ are the anchors of two fan-bundles
  $B_{a_1}$ and $B_{a_2}$ that cross. Note that the crossing can potentially be
  realized by two fan-bundles $B_{b_1}$ and $B_{b_3}$ anchored at $b_1$ and $b_3$, respectively. 
  However, in this case $B_{b_1}$ and $B_{b_3}$ would prevent any connection to $b_2$.
  
\begin{figure}[b]
  \centering
  \includegraphics[page=1]{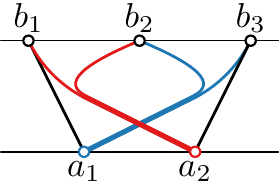}
  \caption{A $1$-sided $2$-layer \fb drawing of $K_{2,3}$.}
  \label{fig:2layer-k23}
\end{figure}
  
  We now prove that the complete bipartite graph $K_{2,4}$ is not $1$-sided $2$-layer \fb.
  Let $\{a_1,a_2\}$ and $\{b_1,b_2,b_3,b_4\}$ be the two partition sets of $K_{2,4}$.
  Since the complete bipartite graph $K_{2,3}$ induced by $a_1$, $a_2$, $b_1$, $b_2$, and $b_3$
  has a unique $1$-sided $2$-layer \fb drawing, it suffices to prove that it is not
  possible to add vertex~$b_4$ to this drawing and to connect it to both~$a_1$ and~$a_2$
  without violating its $1$-sided $2$-layer \fby; we assume as above that $x(a_1)<x(a_2)$ and
  $x(b_1)<x(b_2)<x(b_3)$. 
  By symmetry, we only have to consider two cases: $x(b_4)<x(b_1)$ and $x(b_1)<x(b_4)<x(b_2)$.
  In the first case, $b_4$ cannot be connected to~$a_2$, 
  as this connection would cross two fan-bundles incident to~$a_1$.
  In the second case, $b_4$ cannot be connected to~$a_1$,
  as this connection would cross an unbundled part of an edge incident to~$a_2$. 
\end{proof}

Since Binucci et al.~\cite{DBLP:conf/gd/BinucciCDGKKMT15} showed that a biconnected bipartite graph
is $2$-layer fan-planar if and only if it is a spanning subgraph of a snake,
Lemma~\ref{lem:$2$-layer-k24} immediately leads to a characterization of biconnected $2$-layer 
\fb graphs; see Lemma~\ref{lem:2layer-bicon-baby-snakes}. 
We say that a snake is a \emph{baby snake} if each graph in its chain is a~$K_{2,2}$ or a $K_{2,3}$; see Fig.~\ref{fig:2layer-babysnake} for an example. 

\begin{lemma}\label{lem:2layer-bicon-baby-snakes}
A biconnected bipartite graph is $2$-layer \fblong if and only if it is a spanning subgraph of a baby snake.
\end{lemma}

A direct consequence of the aforementioned characterization is that we can recognize 
(and in the case of an affirmative answer also draw) these graphs, 
by employing the corresponding recognition (and drawing, respectively) algorithm by Binucci et al.~\cite{DBLP:conf/gd/BinucciCDGKKMT15}.
We summarize this observation in the following theorem.

\begin{theorem}\label{thm:2layer-bicon-recognition}
  Biconnected $1$-sided $2$-layer \fblong graphs can be recognized and drawn in linear time.
\end{theorem}

In the remainder of this subsection, we relax biconnectivity and require maximality. Binucci et
al.~\cite{DBLP:conf/gd/BinucciCDGKKMT15} showed that a bipartite graph is maximal
$2$-layer fan-planar if and only if it is a \emph{stegosaurus}, that is, a chain
of snakes that are connected at \emph{common cutvertices}, where each common
cutvertex is incident to exactly two snakes, plus a set of
degree-1 vertices, called \emph{legs}, each of which is attached to a common cutvertex; see Fig.~\ref{fig:2layer-stegosaurus} for an illustration. 
The following lemma has been proven by Binucci
et al.~\cite{DBLP:conf/gd/BinucciCDGKKMT15}, but the proof also works without
modification for our model.

\begin{figure}
  \centering
  \subfloat[\label{fig:2layer-stegosaurus}]{\includegraphics[page=1]{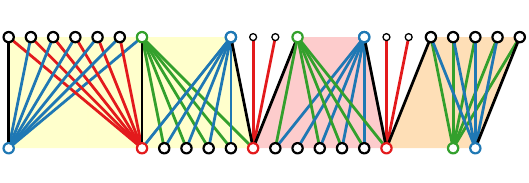}}
  \hfil
  \subfloat[\label{fig:2layer-babystegosaurus}]{\includegraphics[page=2]{2layer-stegosaurus}}
  \caption{Illustration of (a)~a stegosaurus and (b)~a baby stegosaurus.}
  \label{fig:2layer-stegosauri}
\end{figure}

\begin{lemma}[Binucci et al.~\cite{DBLP:conf/gd/BinucciCDGKKMT15}]\label{lem:2layer-independent}
In any $1$-sided $2$-layer \fblong drawing, no biconnected component of a graph can be crossed by an independent edge.
\end{lemma}

By Lemma~\ref{lem:2layer-independent}, it follows that the biconnected components
in a $2$-layer \fb drawing are placed next to each other without crossings.
In the following lemma, we describe the structure of maximal $1$-sided $2$-layer \fb graphs without legs.
To this end, we need the following definition. 
We call a stegosaurus a \emph{baby stegosaurus} if its snakes are baby snakes and
if it contains no legs; see Fig.~\ref{fig:2layer-babystegosaurus} for an example. Note that a baby stegosaurus can be drawn $1$-sided $2$-layer
\fb by just drawing its snakes independently, then connecting them
via their common cutvertices.

\begin{lemma}\label{lem:2layer-baby-stegosaurus}
If we remove the legs of a maximal $1$-sided $2$-layer \fblong graph, then we obtain a baby stegosaurus.
\end{lemma}
\begin{proof}
Since every biconnected component of a maximal $1$-sided $2$-layer \fb graph is a baby snake, the lemma holds as long as 
there exist no bridges. If this is not the case, by Lemma~\ref{lem:2layer-independent} any two
components separated by a bridge are drawn without crossing each other. If
the bridge is planar, then we can connect the two components by another edge
that crosses the bridge. On the other hand, if the bridge is crossed by a fan-bundle, then we can
connect the origin of this fan-bundle to the other component by crossing the bridge.
In both cases, we obtain a contradiction to the graph's maximality.
\end{proof}

Note that in a maximal $2$-layer fan-planar graph, there exist no legs. In fact, Binucci et al.~\cite{DBLP:conf/gd/BinucciCDGKKMT15} showed that a leg contained in a $2$-layer fan-planar graph is incident to a $K_{2,h}$, which in turn can be augmented to a $K_{2,h+1}$ by adding an additional edge without affecting fan-planarity. In our case, however, a~$K_{2,3}$ cannot be augmented to a~$K_{2,4}$ in the presence of a leg (due to Lemma~\ref{lem:$2$-layer-k24}), and therefore Lemma~\ref{lem:2layer-baby-stegosaurus} does not immediately yield a characterization of maximal $1$-sided $2$-layer \fb graphs. So, in the following, we investigate to which vertices the legs of a maximal $1$-sided $2$-layer \fb graph can be attached. 

To this end, let $G$ be a maximal $1$-sided $2$-layer \fb graph and let $\Gamma$ be a $2$-layer \fb drawing of $G$. 
Since by Lemma~\ref{lem:2layer-baby-stegosaurus}, graph $G$ is a baby stegosaurus containing legs. We refer to the $K_{2,2}$ and $K_{2,3}$ subgraphs composing the baby snakes of $G$ as \emph{components} of $G$.    

\begin{figure}[t]
    \centering
    \includegraphics[page=2]{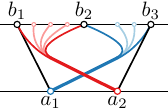}
    \caption{Legs inside a $K_{2,3}$.}
    \label{fig:2layer-k23-legs}
\end{figure}

First, consider a~$K_{2,3}$ component of $G$ with partitions $\{a_1,a_2\}$ and $\{b_1,b_2,b_3\}$ and assume w.l.o.g.~that $x(a_1)<x(a_2)$ and $x(b_1)<x(b_2)<x(b_3)$ in drawing $\Gamma$; see Fig.~\ref{fig:2layer-k23}. There are no legs attached to $b_1$, $b_2$
and~$b_3$ that lie between~$a_1$ and~$a_2$ in $\Gamma$, because the interval between $a_1$ and~$a_2$ is ``blocked'' by the fan-bundles anchored at~$a_1$ and~$a_2$. However, $G$ may have any number of legs attached to~$a_1$ and~$a_2$ that lie between $b_2$ and $b_3$, and between $b_1$ and $b_2$, respectively; see Fig.~\ref{fig:2layer-k23-legs}. 

In the following lemma, we focus on legs attached to vertices of a~$K_{2,2}$ component of $G$.

\begin{lemma}\label{lem:2layer-k22-noleg}
There exist no leg in $G$ that is attached to a vertex that belongs to a $K_{2,2}$ component of $G$.
\end{lemma}

\begin{proof}
Consider a~$K_{2,2}$ component of $G$ with partition sets $\{a_1,a_2\}$ and $\{b_1,b_2\}$, such that~$x(a_1)<x(a_2)$
and $x(b_1)<x(b_2)$ holds in $\Gamma$; see Figs.~\ref{fig:2layer-k22-legs-1} and~\ref{fig:2layer-k22-legs-2}.
Note that any $2$-layer drawing of $K_{2,2}$ is not crossing-free. 
Hence, it must contain two fan-bundles that cross. 
Assume w.l.o.g.~that one fan-bundle is anchored at~$a_1$. 
Then, the other fan-bundle is either anchored at~$a_2$ or at~$b_1$, since anchoring it at~$b_2$ is not possible. 

\begin{figure}[b]
	\centering
	\subfloat[\label{fig:2layer-k22-legs-1}]{\includegraphics[page=1]{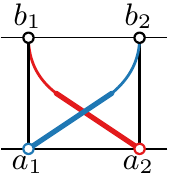}}
	\hfil
	\subfloat[\label{fig:2layer-k22-legs-2}]{\includegraphics[page=2]{2layer-k22-legs}}
	\hfil
	\subfloat[\label{fig:2layer-k22-legs-3}]{\includegraphics[page=3]{2layer-k22-legs}}
	\hfil
	\subfloat[\label{fig:2layer-k22-legs-4}]{\includegraphics[page=4]{2layer-k22-legs}}
	\caption{(a)-(b) The two ways to attach the fan-bundles inside a $K_{2,2}$,
    (c) all legs are attached to $a_1$ and $a_2$, and
    (d) all legs are attached to $a_1$ and $b_1$.}
\end{figure}

We will first prove that there exist no leg in $G$ that lies between $a_1$ and $a_2$, or between $b_1$ and $b_2$ in $\Gamma$. 
For a proof by contradiction, assume that there exists such a leg. 
We first consider the case in which the two fan-bundles are anchored at $a_1$ and $a_2$; see Fig.~\ref{fig:2layer-k22-legs-3}. 
In this case, there exists no leg incident to $b_1$ and~$b_2$ that lies between $a_1$ and $a_2$.
Hence, there exists at least one leg attached to either $a_1$ or $a_2$, say to the former, that lies between $b_1$ and~$b_2$ in~$\Gamma$.
Then, the maximality of $G$ is contradicted, as it is possible to add an edge between the leftmost such leg and $a_2$ without violating the \fby of~$\Gamma$.

We now consider the case in which the two fan-bundles are anchored at $a_1$ and $b_1$; see Fig.~\ref{fig:2layer-k22-legs-4}.
Observe that there exists no leg in $G$ attached to $a_2$ or $b_2$ that lies between $b_1$ and $b_2$, and between $a_1$ and $a_2$, respectively.
We can further assume w.l.o.g.~that each of $a_1$ and $b_1$ has at least one leg that lies between $b_1$ and $b_2$, and between $a_1$ and $a_2$, respectively. In fact, if one of these two vertices, say $b_1$, does not have any such leg, then we may assume that the second fan-bundle is not anchored at $b_1$ but at $a_2$, which is a case that has already been considered.     

First, assume $b_1$ and $b_2$ are the only neighbors of $a_1$ that belong to some component of $G$. 
Then, either $b_1$ is a cutvertex or $a_1$ is the leftmost vertex on its layer that is not a leg of $b_1$. 
In both cases, we move the legs of $b_1$ to the left of $a_1$; see Fig.~\ref{fig:2layer-a1-noleg-1}.
Then, $b_1$ has no leg that lies between $a_1$ and $a_2$, which contradicts our assumption. 

\begin{figure}[t]
	  \subfloat[\label{fig:2layer-a1-noleg-1}]{\includegraphics[scale=0.9,page=1]{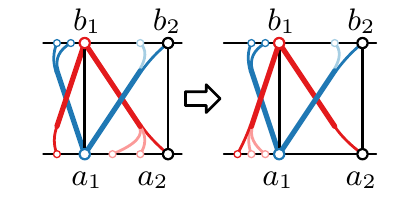}}
  	  \hfil
	  \subfloat[\label{fig:2layer-a1-noleg-2}]{\includegraphics[scale=0.9,page=2]{2layer-a1-noleg}}
    \hfil
	  \subfloat[\label{fig:2layer-a1-noleg-3}]{\includegraphics[scale=0.9,page=3]{2layer-a1-noleg}}
	  \caption{Proof that there is no leg attached to $a_1$ that lies between $b_1$ and $b_2$:
	  (a)~$b_1$ and $b_2$ are the only neighbors of $a_1$ that belong to some component of $G$,
	  (b)~$a_1$ and $b_1$ belong to a $K_{2,3}$ component, and
	  (c)~$a_1$ and $b_1$ belong to a second $K_{2,2}$ component.}
	  \label{fig:2layer-a1-noleg}
\end{figure}

Now, consider the case that $b_1$ and $b_2$ are not the only neighbors of $a_1$ that belong to some component of $G$.
Assume first that $a_1$ and $b_1$ also belong to a~$K_{2,3}$ component; see Fig.~\ref{fig:2layer-a1-noleg-2}.
Assume w.l.o.g.~that $a_1$ belongs to the partition set of the $K_{2,3}$ components containing the two vertices.
We move the legs of~$a_1$ that lie between $b_1$ and $b_2$ inside the $K_{2,3}$ component; see Fig.~\ref{fig:2layer-a1-noleg-2}. 
Then, as before, $a_1$ has no leg that lies between $b_1$ and $b_2$, which contradicts our assumption.
It remains to consider the case in which $a_1$ and $b_1$ belong to a second $K_{2,2}$ component; see Fig.~\ref{fig:2layer-a1-noleg-3}.
Let $a_2'$ and $b_2'$ be the additional vertices of this $K_{2,2}$ component.
If the only legs inside the second $K_{2,2}$ component, if any, are also attached only to $a_1$ and $b_1$, 
then we can move all the legs attached to $a_1$ inside the first $K_{2,2}$ component and all the legs attached to $b_1$ inside the second $K_{2,2}$ component, which again contradicts our assumption; see Fig.~\ref{fig:2layer-a1-noleg-3}.
It follows that the legs of the second $K_{2,2}$ component must be attached only to~$a_2'$ and $b_2'$.

By applying the above arguments for the second $K_{2,2}$ component, we either obtain a contradiction or we conclude that $a_2'$ and $b_2'$ belong to a third $K_{2,2}$ component with the same properties. By repeating the same argument, we will eventually obtain a chain of $K_{2,2}$ components, all with the same properties. At the end of this chain, there must be either a cutvertex, or a~$K_{2,3}$ component, or the leftmost (or the rightmost) vertex of one of the two layers. Thus, one of the previous cases applies in order to derive a contradiction. This completes the proof that there is no leg between $a_1$ and $a_2$, or between $b_1$ and $b_2$.     

To conclude the proof of the lemma, assume that there is a leg attached to a vertex of our $K_{2,2}$ component, say $a_1$, that does not lie between $b_1$ and $b_2$. Since there are legs neither between $a_1$ and $a_2$ nor between $b_1$ and $b_2$, we can move this leg between $b_1$ and $b_2$, which yields one of the cases that we have already considered before. 
\end{proof}

From the above discussion, it follows that a leg can only be attached to
\begin{enumerate}[noitemsep,label=(\roman*)]
\item the leftmost or rightmost vertex of a snake that is not a common cutvertex, if it belongs to a~$K_{2,3}$ component, or
\item a common cutvertex that belongs to two~$K_{2,3}$ components, or
\item a vertex that belongs to two~$K_{2,3}$ components, and in each of them belongs to the partition set containing two vertices.
\end{enumerate}

\begin{figure}
  \centering
  \includegraphics[page=3]{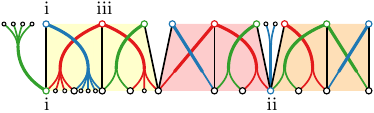}
  \caption{The baby stegosaurus from Fig.~\ref{fig:2layer-babystegosaurus} with big legs. Each eligible vertex has at least one leg and the letters correspond to the cases that allow a leg to exist.}
  \label{fig:2layer-babystegosaurusbiglegs}
\end{figure}

We refer to such a leg as \emph{big leg}; see Fig.~\ref{fig:2layer-babystegosaurusbiglegs} for an example. This gives rise to the following simple recognition and drawing algorithm.

\begin{theorem}
Maximal $1$-sided $2$-layer \fblong graphs can be recognized and drawn in linear time.
\end{theorem}
\begin{proof}
We first remove all legs. 
By Lemma~\ref{lem:2layer-baby-stegosaurus}, the resulting graph has to be a baby stegosaurus. 
We split it at its cutvertices and use the recognition and drawing algorithm for snakes by Binucci et al.~\cite{DBLP:conf/gd/BinucciCDGKKMT15}.
For each of the snakes, we can easily check whether it is a baby snake. 
In the negative case, we reject the instance. 
Otherwise, we glue the baby snakes together at their cutvertices. 
By Lemma~\ref{lem:2layer-k22-noleg}, we only have to check whether the legs that we removed at the begining of our algorithm are big legs, which can be done in linear time. 
In the negative case, we reject the instance.
Otherwise, we draw each of them either between the two $K_{2,3}$ components it belongs to 
or at the leftmost or rightmost vertex on one of the layers, if this belongs to a~$K_{2,3}$ component.
\end{proof}

\subsection{Triconnected $1$-sided outer-\fblong graphs.}
\label{subsec:recog:outer}

In this section, we present a linear-time algorithm for the
recognition of triconnected $1$-sided outer-\fb graphs, which in the
case of a positive instance also computes a corresponding $1$-sided
outer-\fb drawing. To do so, we will first present some important
properties of triconnected $1$-sided outer-\fb graphs. We start with
a property of biconnected (and hence of triconnected) $1$-sided
outer-\fb graphs.


\begin{lemma}\label{lem:3con-nonplanar}
Let $G$ be a biconnected $1$-sided outer-\fblong graph. 
A $1$-sided outer-\fblong drawing $\Gamma$ of $G$ can be augmented (by adding edges) 
into a $1$-sided outer-\fblong drawing $\Gamma'$ in which all edges on
the outer face of $\Gamma'$ are planar.
%
\end{lemma}
\begin{proof}
Let $\Pi$ be the \emph{planarization} of drawing~$\Gamma$, i.e.,
$\Pi$ is the drawing obtained by replacing the crossing points of
$\Gamma$ with dummy vertices. Since $G$ is biconnected, the outer
face of~$\Pi$ is a simple cycle and contains two types of vertices;
vertices of~$G$ and vertices that correspond to crossing points
of~$\Gamma$. Let $v_{1},\ldots,v_{k}$ be the vertices that are
incident to the outer face of $\Pi$ as they appear in clockwise order
along it. Since~$\Gamma$ is outer-\fb, the outer face of $\Pi$
contains all vertices of $G$.

Note that if the outer face of $\Pi$ consists exclusively of vertices
of $G$, then the lemma clearly holds. In particular, for any two
vertices $v_i$ and $v_{i+1}$ of $G$ that are consecutive along the
outer face of $\Pi$, the edge $(v_i,v_{i+1})$ belongs to $G$; see
Fig.~\ref{fig:outer-3-planar-1}. To complete the proof, assume that
there exists a vertex, say $v_i$, along the outer face of $\Pi$ that
corresponds to a crossing point in $\Gamma$. By outer-\fby, it
follows that $v_{i-1}$ and $v_{i+1}$ are both vertices of $G$.  We
remove~$v_i$ from the outer face of $\Pi$ as follows. If the
edge~$(v_{i-1},v_{i+1})$ exists in~$G$, we remove it from drawings
$\Pi$ and $\Gamma$. Then, we add the edge~$(v_{i-1},v_{i+1})$
to~$\Pi$ and to~$\Gamma$ as a curve that starts in~$v_{i-1}$, follows
the outer face until~$v_i$, and ends in~$v_{i+1}$ again by following
the outer face; see Fig.~\ref{fig:outer-3-planar-2}.

\begin{figure}[t]
	\centering
    \subfloat[\label{fig:outer-3-planar-1}]{\includegraphics[page=3]{outer-3-planar}}
    \hfil
    \subfloat[\label{fig:outer-3-planar-2}]{\includegraphics[page=1]{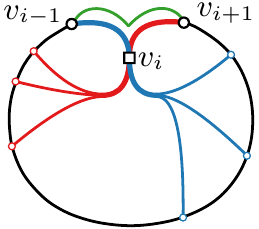}}
    \caption{Creating an outer-\fb drawing in which all edges of its outer face are planar.}
    \label{fig:outer-3-planar}
\end{figure}

Note that the aforementioned procedure does not
reduce the number of vertices of $G$ that are on the outer face of
either $\Pi$ or $\Gamma$, which implies that if we apply this
procedure iteratively to each vertex of $\Pi$ that corresponded to a
crossing point of $\Gamma$, we will eventually obtain a drawing
$\Gamma'$ in which all edges of its outer face are planar.
\end{proof}

\begin{lemma}\label{lem:3con-onecrossing}
  The following properties hold in a $1$-sided outer-\fb drawing $\Gamma$ of a triconnected graph $G$ in which
  all edges incident to its outer face are planar: %
  \begin{enumerate}[label=P.\arabic*]
    \item\label{prop:outer-3-crossed} No inner edge of $\Gamma$ is planar.
    \item\label{prop:outer-3-adjacent} The anchors of two crossing fan-bundles in $\Gamma$ are consecutive along the outer face of $\Gamma$.
    \item\label{prop:outer-3-onecrossing} There is at most one fan-bundle crossing in $\Gamma$.
  \end{enumerate}
\end{lemma}
\begin{proof}
If there is an inner edge~$(u,v)$ that is planar in~$\Gamma$, then~$u$ and~$v$ form a separation pair in $G$. Since $G$ is triconnected, this is a contradiction. Hence, Property~\ref{prop:outer-3-crossed} holds. 
  
Let $B_u$ and $B_v$ be two crossing fan-bundles in $\Gamma$ that are anchored 
at vertices~$u$ and $v$ of $G$. To prove Property~\ref{prop:outer-3-adjacent}, 
assume to the contrary that $u$ and~$v$ are not consecutive along the outer 
face of $\Gamma$, i.e., $(u,v)$ is not an edge of the outer face of $\Gamma$. 
We proceed by drawing edge $(u,v)$ in $\Gamma$ as a $B_uB_v$-following curve 
(and hence planar), by first removing it from $\Gamma$ in case that it belongs 
to $G$. Since $u$ and $v$ are not consecutive along the outer face of $\Gamma$, 
it follows that~$u$ and~$v$ form a separation pair in $G$. Since $G$ is 
triconnected, this is a contradiction and Property~\ref{prop:outer-3-adjacent}
 holds.

To prove Property~\ref{prop:outer-3-onecrossing}, assume for a contradiction that there exist two fan-bundle crossings in $\Gamma$, say between fan-bundles~$B_u$ and $B_v$ and between fan-bundles~$B_w$ and $B_z$, respectively. Clearly, $u \neq v$ and $w \neq z$ hold. Let $u_1,\ldots,u_\kappa$ and $v_1,\ldots,v_\lambda$ be the tips of $B_u$ and $B_v$, respectively, in this clockwise order along the outer face of $\Gamma$. Accordingly, let $w_1,\ldots,w_\mu$ and $z_1,\ldots,z_\nu$ be the tips of $B_w$ and $B_z$, respectively, in this clockwise order along the outer face of $\Gamma$. We now claim that $u \notin \{w,z\}$ and $v \notin \{w,z\}$ holds. Assume to the contrary that $v=w$. Then, by Property~\ref{prop:outer-3-adjacent}, we may further assume w.l.o.g.~that $u$, $v$ and $z$ appear consecutively in this clockwise order along the outer face of $\Gamma$; see Fig.~\ref{fig:outer-3-onecrossing-1}. In this case, however, either $\langle v, z_\nu \rangle$ or $\langle v, u_1 \rangle$ form a separation pair in $G$, which is a contradiction to the fact that $G$ is triconnected. So, we may assume that $u$, $v$, $w$ and $z$ are pairwaise disjoint and w.l.o.g.~that they appear in this clockwise order along the outerface of $\Gamma$; see Fig.~\ref{fig:outer-3-onecrossing-2}. In this case, however, $\langle w, z_\nu \rangle$ or $\langle v, u_1 \rangle$ form a separation pair in $G$, which is again a contradiction to the fact that $G$ is triconnected. Hence, Property~\ref{prop:outer-3-onecrossing} holds.
\end{proof}
  
\begin{figure}[t]
    \centering
    \subfloat[\label{fig:outer-3-onecrossing-1}]{\includegraphics[page=1]{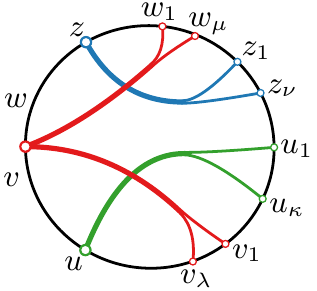}}
    \hfil
    \subfloat[\label{fig:outer-3-onecrossing-2}]{\includegraphics[page=2]{outer-3-onecrossing}}    
    \caption{Illustration of the case of two fan-bundle crossings in Property~\ref{prop:outer-3-onecrossing}.}
    \label{fig:outer-3-onecrossing}
\end{figure}  

We call a drawing with Properties~\ref{prop:outer-3-crossed}, \ref{prop:outer-3-adjacent} and \ref{prop:outer-3-onecrossing} of Lemma~\ref{lem:3con-onecrossing} a \emph{canonical drawing}. In the following, we give a complete characterization of the triconnected $1$-sided outer-\fb graphs.

\begin{lemma}\label{lem:3con-characterization}
  A triconnected graph $G$ with~$n\ge 5$ vertices is $1$-sided outer-\fblong if and
  only if it consists of:
  \begin{enumerate}[label=C.\arabic*]
  \item\label{c1} a Hamiltonian path~$v_1,v_2,\ldots,v_n$,
  \item\label{c5} the edges~$(v_1,v_{n-1})$ and~$(v_n,v_2)$,
  \item\label{c2} the edges~$(v_n,v_i)$, with~$3\le i\le k-1$, and~$(v_1,v_j)$, with~$k\le j\le n-2$ for some~$2\le k\le n$,
  \item\label{c3} the edge~$(v_1,v_n)$ if $k\in\{2,n-1\}$, and
  \item\label{c4} possibly the edges~$(v_n,v_k)$ and~$(v_1,v_n)$.
  \end{enumerate}
\end{lemma}
\begin{proof}
  For the sufficiency part, in order to prove the triconnectivity, we show that there are at least~3 vertex-disjoint paths between each pair of vertices~$u$ and~$v$ of $G$. If~$\{u,v\}=\{v_1,v_n\}$, then there exist paths $v_1,v_2,v_n$ (by~\ref{c1} and~\ref{c5}) and $v_1,v_{n-1},v_n$ (by~\ref{c2}). For the third path, we choose the path $v_1,v_k,v_{k-1},v_n$ if $3<k<n-1$ (which exists by~\ref{c2}), the path $v_1,v_{k+1},v_k,v_n$ if $k=3$ (which exists by~$n\ge 5$), and the path $v_1,v_n$ if $k\in\{2,n-1\}$ (which exists by~\ref{c3}). 
Consider now a pair $v_i$ and $v_j$ of vertices, with $i < j$. Assume that~$i<k$; the other case is symmetric. This implies that the edge $(v_n,v_i)$ exists (by~\ref{c2}). If~$j<k$, then the edge $(v_n,v_j)$ exists (by~\ref{c2}) and $v_i$ and $v_j$ are connected by the path $v_i,v_{i+1},\ldots, v_j$, the path $v_i,v_n,v_j$, and the path $v_i,v_{i-1},\ldots,v_1,v_k,v_{k-1},\ldots,v_j$. If $j\ge k$, then the edge $(v_1,v_j)$ exists (by~\ref{c2}), and thus $v_i$ and $v_j$ are connected by the path $v_i,v_{i+1},\ldots, v_j$, the path $v_i,v_n,v_{n-1},\ldots,v_j$, and the path $v_i,v_{i-1},\ldots,v_1,v_j$. This proves the triconnectivity.
Fig.~\ref{fig:outer-3-drawing} is an evidence that such a graph always admits a $1$-sided outer-\fb drawing.

For the necessity, first assume that~$G$ is maximal $1$-sided outer-\fb. By Lemma~\ref{lem:3con-nonplanar}, there is a $1$-sided outer-\fb drawing~$\Gamma$ whose outer face is a simple planar Hamiltonian cycle~$v_1,\ldots,v_n,v_1$. Thus,~\ref{c1} holds. Since $G$ is triconnected, there is at least one inner edge. By Lemma~\ref{lem:3con-onecrossing}, $\Gamma$ is canonical, and hence there exist exactly two crossing fan-bundles in $\Gamma$, whose origins $v_n$ and~$v_1$ are adjacent along the outer face of $\Gamma$. 
Since every vertex of $G$ has degree at least~3, each vertex $v_2,\ldots,v_{n-1}$ is a tip of~$B_{v_n}$ or~$B_{v_1}$. Also, since the edges belonging to~$B_{v_n}$ and~$B_{v_1}$ cannot be further crossed, it follows that the tips of~$B_{v_n}$ and of~$B_{v_1}$ form two interior-disjoint intervals along the outer face of $\Gamma$. Hence, all edges described by \ref{c5} and~\ref{c2} belong to $G$; the edges from~\ref{c5} have to exist since~$v_2$ can only be a tip of $B_{v_n}$ and
since~$v_{n-1}$ can only be a tip of~$B_{v_1}$ by simplicity.
Since $G$ is maximal, there exists a vertex $v_k$, with $2 \leq k \leq n-1$, that is a tip of both $B_{v_n}$ and~$B_{v_1}$. Hence, \ref{c3} and \ref{c4} hold. 

\begin{figure}[t]
  \centering
  \includegraphics[page=1]{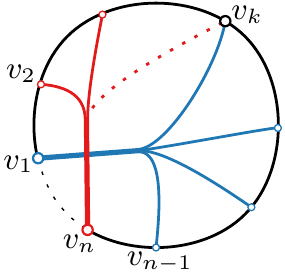}
  \caption{A $1$-sided outer-\fb canonical drawing of a triconnected graph.}
  \label{fig:outer-3-drawing}
\end{figure}

Assume now that~$G$ is not maximal. By Lemma~\ref{lem:3con-nonplanar}, $G$ is a subgraph of a maximal $1$-sided outer-\fb graph $G'$. The only edges that can be removed from $G'$ without violating triconnectivity are one of $(v_1,v_k)$ and $(v_n,v_k)$, plus the edge~$(v_1,v_n)$, but only if $k\notin\{2,n-1\}$; otherwise, at least one of~$v_1$ and~$v_n$ has degree smaller than~3. Hence, \ref{c2}, \ref{c3} and \ref{c4} hold. This also implies that the Hamiltonian cycle $v_1,\ldots,v_n,v_1$ can become a Hamiltonian path $v_1,\ldots,v_n$, and thus \ref{c1} still holds. This concludes the proof of the lemma.
\end{proof}
Based on Lemma~\ref{lem:3con-characterization}, we can derive a linear-time
recognition algorithm for triconnected graphs, which in the case of a positive instance also computes a $1$-sided outer-\fb drawing. To this end, we have to find the Hamiltonian
path~$v_1,\ldots,v_n$. While it is NP-hard in general to find a Hamiltonian path,
we show that we can do so efficiently for this graph class by identifying the
vertices~$v_1$ and~$v_n$ based on their degree.

\newcommand{\outerthreerecognitionText}{Triconnected $1$-sided outer-\fblong graphs can be recognized and drawn in linear time.}
\begin{theorem}\label{thm:outer-3-recognition}
\outerthreerecognitionText
\end{theorem}
\begin{proof}
Let $G$ be any triconnected graph. The task is to test whether $G$ satisfies the conditions \ref{c1}-\ref{c4} of Lemma~\ref{lem:3con-characterization}. Note that, in order for these conditions to be satisfied, the only vertices that can have degree larger than~$3$ are $v_n$,~$v_1$, and~$v_k$. In particular, the sum of the degrees of~$v_1$ and~$v_n$ is between~$n$ and~$n+3$,~$v_k$ has degree at most~$4$, while every other vertex has degree exactly~$3$. More specifically, the sum of the degrees of~$v_1$ and~$v_n$ is 
\begin{itemize}
\item $n$, if $\deg v_k=3$ and edge $(v_1,v_n)$ does not belong to $G$;
\item $n+1$, if $\deg v_k=4$ and edge $(v_1,v_n)$ does not belong to $G$;
\item $n+2$, if $\deg v_k=3$ and edge $(v_1,v_n)$ belongs to $G$;
\item $n+3$, if $\deg v_k=4$ and edge $(v_1,v_n)$ belongs to $G$.
\end{itemize}
Hence, our algorithm rejects $G$ if one of the following holds:
\begin{itemize}
  \item there are more than three vertices with degree larger than~$3$;
  \item there are more than two vertices with degree larger than~$4$;
  \item there are no two vertices such that the sum of their degrees is between~$n$ and~$n+3$.
\end{itemize}

Fig.~\ref{fig:outer-3-recognition} shows all triconnected $1$-sided outer-\fb
graphs with~$n\le 8$. So, if~$G$ has at most eight vertices, our algorithm exhaustively tests whether $G$ is one of these graphs. So, we may assume w.l.o.g.~that $G$ has at least nine vertices. Then, the sum of the degrees of~$v_1$ and~$v_n$ must be at least~$9$, so there always exists at least one vertex of degree larger than~$4$. We distinguish the following cases. Table~\ref{tab:outer-3-recognition-cases} gives an overview of the cases and shows that the case analysis is complete.

\begin{figure}[t]
  \centering
  \begin{tabular}{c@{\qquad}c@{\qquad}c@{\qquad}c@{\qquad}c@{\qquad}}
    \includegraphics[page=1]{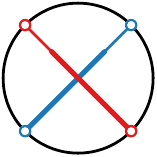}
    &
    \includegraphics[page=2]{outer-3-recognition}
    &
    \includegraphics[page=3]{outer-3-recognition}
    &
    \includegraphics[page=4]{outer-3-recognition}
    &
    \includegraphics[page=5]{outer-3-recognition}
    \\[.5em]
    \includegraphics[page=6]{outer-3-recognition}
    &
    \includegraphics[page=7]{outer-3-recognition}
    &
    \includegraphics[page=8]{outer-3-recognition}
    &
    \includegraphics[page=9]{outer-3-recognition}
    &
    \includegraphics[page=10]{outer-3-recognition}
    \\[.5em]
    \includegraphics[page=11]{outer-3-recognition}
    &
    \includegraphics[page=12]{outer-3-recognition}
    &
    \includegraphics[page=13]{outer-3-recognition}
    &
    \includegraphics[page=14]{outer-3-recognition}
    &
    \includegraphics[page=15]{outer-3-recognition}
    \\[.5em]
    \includegraphics[page=16]{outer-3-recognition}
    &
    \includegraphics[page=17]{outer-3-recognition}
    &
    \includegraphics[page=18]{outer-3-recognition}
    &
    \includegraphics[page=19]{outer-3-recognition}
    &
    \includegraphics[page=20]{outer-3-recognition}
    \\[.5em]
    \includegraphics[page=21]{outer-3-recognition}
    &
    \includegraphics[page=22]{outer-3-recognition}
    &
    \includegraphics[page=23]{outer-3-recognition}
    &
    \includegraphics[page=24]{outer-3-recognition}
    &
    \includegraphics[page=25]{outer-3-recognition}
    \\[.5em]
    \includegraphics[page=26]{outer-3-recognition}
    &
    \includegraphics[page=27]{outer-3-recognition}
    &
    \includegraphics[page=28]{outer-3-recognition}
    &
    \includegraphics[page=29]{outer-3-recognition}
    &
    \includegraphics[page=30]{outer-3-recognition}
  \end{tabular}
    \caption{All triconnected $1$-sided outer-\fb graphs with at most~8 vertices.}
    \label{fig:outer-3-recognition}
\end{figure}

\newcolumntype{M}[1]{>{\centering\arraybackslash}m{#1}}
\begin{table}[t]
\centering
\caption{An illustration of all cases in the proof of Theorem~\ref{thm:outer-3-recognition}
with $n\ge 9$ that can occur by $n\le \deg v_n+\deg v_1 \le n+3$ and $3\le \deg v_k\le 4$,
assuming that~$\deg v_n>4\ge\deg v_1$.
The remaining case that~$\deg v_n\ge \deg v_1>4$ is handled in Case~\ref{c:4-4}.
Columns marked by an X cannot occur because of $\deg v_n+\deg v_1\ge n$.}
\label{tab:outer-3-recognition-cases}
\begin{tabular}{M{2.5cm}M{0cm}|M{0cm}M{1.8cm}M{1.8cm}M{0cm}M{1.8cm}M{1.8cm}M{0cm}@{}m{0pt}@{}}
\toprule
            &\multicolumn{1}{c}{}&& \multicolumn{2}{c}{$\deg v_1=4$}                 && \multicolumn{2}{c}{$\deg v_1=3$}                 &&\\
            \cmidrule{4-5}\cmidrule{7-8}
            &\multicolumn{1}{c}{}&& $\deg v_k=4$ & $\deg v_k=3$ && $\deg v_k=4$ & $\deg v_k=3$ &&\\
            \midrule
$\deg v_n=n-1$ &&& Case \ref{c:n-4-4-4}                 & Case \ref{sc:n-1-4}                  && Case \ref{sc:n-1-4}                  & Case \ref{sc:n-1-3}                  &&\\[5pt]
$\deg v_n=n-2$ &&& Case \ref{c:n-4-4-4}                 & Case \ref{ssc:n-2-4-v1}            && Case \ref{ssc:n-2-4-vk}            & Case \ref{sc:n-2-3}                  &&\\[5pt]
$\deg v_n=n-3$ &&& Case \ref{c:n-4-4-4}                 & Case \ref{sc:n-1-4}                  && Case \ref{sc:n-1-4}                  & Case \ref{sc:n-3-3}                  &&\\[5pt]
$\deg v_n=n-4$ &&& Case \ref{c:n-4-4-4}                 & Case \ref{sc:n-4-4}                  && X                                 & X       &&\\             \bottomrule
\end{tabular}
\end{table}

\ccase{c:4-4} \emph{There are two vertices of $G$ with degree larger than~$4$}.
  These two vertices must be~$v_1$ and~$v_n$. If we remove them from $G$,
  what remains must be a path~$v_2,\ldots,v_{n-1}$, which prescribes the Hamiltonian
  path together with~$(v_1,v_2)$ and $(v_{n-1},v_n)$;
  see Fig.~\ref{fig:outer-5-degrees-1}.

\ccase{c:n-3-3} \emph{There is a vertex of $G$ with degree at least~$n-3$, and every other
  vertex has degree~$3$}. We label the high-degree vertex as~$v_n$, and we distinguish three subcases, based on its degree.

\begin{changemargin}{.2in}{0in}
 	
\subcase{sc:n-1-3} $\deg(v_n)=n-1$. Then,~$v_n$ is connected to all other vertices.
  Hence, if we remove~$v_n$ from~$G$, what remains must be a cycle. So, we can choose any vertex as~$v_1$ and one of its incident
  edges as~$B_{v_1}$; see Fig.~\ref{fig:outer-5-degrees-2}.

\subcase{sc:n-2-3} $\deg(v_n)=n-2$. Then, the sum of the degrees over all vertices of $G$ is 
$n-2+3\cdot (n-1) = 4n-5$, which is not possible, since the sum of the degrees over all vertices of a graph is always even.

\subcase{sc:n-3-3} $\deg(v_n)=n-3$. Then, $\deg(v_1)+\deg(v_n)=n$, so the
  edge~$(v_1,v_n)$ does not exist. Hence, the three edges from~$v_1$ are~$(v_1,v_2)$,
  $(v_1,v_{n-2})=(v_1,v_k)$, and~$(v_1,v_{n-1})$. Since \mbox{$\deg(v_k)=3$}, the edge
  $(v_k,v_n)$ does not exist, so~$v_n$ is connected to every vertex except for~$v_1$
  and~$v_k$. We label as~$v_1$ one of the two vertices that are not connected to~$v_n$.
  If we now remove~$v_1$ and~$v_n$ from~$G$, what remains must be again a
  path~$v_2,\ldots,v_{n-1}$, which prescribes the Hamiltonian
  path together with~$(v_1,v_2)$ and $(v_{n-1},v_n)$; see Fig.~\ref{fig:outer-5-degrees-3}.

\end{changemargin}

\ccase{c:n-2-4} \emph{There is a vertex of $G$ with degree at least~$n-4$, one vertex with
  degree~4, and each other vertex has degree~3}.
  We label the high-degree vertex as~$v_n$, and we distinguish again three subcases.

\begin{changemargin}{.2in}{0in}

\subcase{sc:n-1-4} $\deg(v_n)=n-1$ or $\deg(v_n)=n-3$. In the former case,
  the sum of degrees over all vertices of $G$ is $n-1 + 4 + 3\cdot (n-2) = 4n-3$, while in the latter case it
  is~$4n-5$; since in both cases this value is odd, we conclude that this case cannot occur.

\subcase{sc:n-2-4} $\deg(v_n)=n-2$. In this case, we do not know
  whether~$v_1$ or~$v_k$ is the degree-4 vertex, so we have to try both
  possibilities.
\end{changemargin}

\begin{changemargin}{.4in}{0in}

\subsubcase{ssc:n-2-4-v1} We label the degree-4 vertex as~$v_1$.
  Then, we have a similar situation as in Case~\ref{c:4-4}, with the only difference being that the edge~$(v_k,v_n)$ does not exist. We thus proceed as in this case; see Fig.~\ref{fig:outer-5-degrees-3}.

\subsubcase{ssc:n-2-4-vk} We label the degree-4 vertex as~$v_k$.
  Then,~$v_1$ has degree~3 and we have~$\deg(v_1)+\deg(v_n)=n+1$,
  so~$v_1$ has to be the only vertex not adjacent to~$v_n$.
  If we now remove~$v_n$ and~$v_1$ from~$G$, what remains must be again a
  path~$v_2,\ldots,v_{n-1}$, which prescribes the Hamiltonian path
  together with~$(v_1,v_2)$ and $(v_{n-1},v_n)$; see
  Fig.~\ref{fig:outer-5-degrees-3}.

\end{changemargin}

\begin{changemargin}{.2in}{0in}

\subcase{sc:n-4-4} $\deg(v_n)=n-4$.
  Since~$\deg(v_1)+\deg(v_n)\ge n$, we have that~$v_1$ must be the degree-4 vertex and the
  edge~$(v_1,v_n)$ does not exist. Then, since~$\deg(v_k)=3$, the edge~$(v_k,v_n)$
  also does not exist. Hence, the situation is the same as in Case~\ref{c:4-4}
  (with these two edges missing); see Fig.~\ref{fig:outer-5-degrees-1}.
  Hence, we can again prescribe the Hamiltonian path by removing~$v_1$ and~$v_n$.

\end{changemargin}

\ccase{c:n-4-4-4} \emph{There is a vertex of $G$ with degree at least~$n-4$, two vertices with
  degree~$4$, and every other vertex has degree~$3$}. We label the high-degree
  vertex as~$v_n$. One of the vertices with degree~4 has to be~$v_1$, the other
  one has to be~$v_k$. In any case,~$v_k$ will have degree~4, so both the
  edges~$(v_1,v_k)$ and~$(v_k,v_n)$ must exist. We distinguish two subcases based on whether one or both these degree-$4$ vertices are connected to $v_n$.

\begin{changemargin}{.2in}{0in}

\subcase{sc:n-4-4-4-noedge} One of the two degree-4 vertices is not connected
  to~$v_n$. Since the edge~$(v_k,v_n)$ exists, this vertex must be~$v_1$.
  We can handle this case in the same way as Case~\ref{sc:n-4-4}; see
  Fig.~\ref{fig:outer-5-degrees-1}.

\subcase{sc:n-4-4-4-edge} Both degree-4 vertices are connected to~$v_n$.
  In this case, the edge~$(v_1,v_n)$ exists.
  Since~$v_1$ has degree~$4$, it has two inner edges:~$(v_1,v_{n-1})$
  and~$(v_1,v_{n-2})$ with its other edges being~$(v_1,v_2)$ and~$(v_1,v_n)$.
  This implies that~$k=n-2$, so~$v_k$ has edges~$(v_1,v_k)$,~$(v_k,v_n)$, $(v_
  k,v_{n-1})$, and~$(v_k,v_{n-3})$; thus $v_1$ and $v_k$ only differ in one edge. In fact, by
  removing~$v_n$ and its incident edges we obtain a cycle with the single
  chord~$(v_1,v_k)$, so the whole graph is symmetric and we can choose either of the
  degree-$4$ vertices as~$v_1$; see Fig.~\ref{fig:outer-5-degrees-3}.
 	
\end{changemargin}

  \begin{figure}[t]
    \centering
    \subfloat[\label{fig:outer-5-degrees-1}]{\includegraphics[page=1]{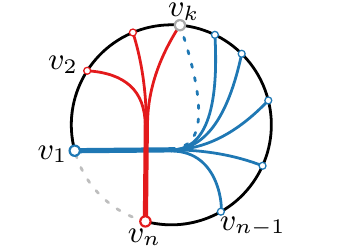}}
    \hfil
    \subfloat[\label{fig:outer-5-degrees-2}]{\includegraphics[page=2]{outer-3-degrees}}
    \hfil
    \subfloat[\label{fig:outer-5-degrees-3}]{\includegraphics[page=3]{outer-3-degrees}}
    \caption{Illustration of the different cases for~$n\ge9$. Dotted edges might
      be there or not, depending on the case:
      (a)~Cases \ref{c:4-4}, \ref{sc:n-4-4}, and~\ref{sc:n-4-4-4-noedge},
      (b)~Case \ref{sc:n-1-3}, and
      (c)~Cases~\ref{sc:n-3-3},~\ref{sc:n-2-4},~and~\ref{sc:n-4-4-4-edge}}
    \label{fig:outer-5-degrees}
  \end{figure}

This completes the description of our recognition algorithm. We can find vertices~$v_1$ and~$v_n$
and the Hamiltonian cycle in linear time and we can check whether the
correct edges are in the graph in linear time as well, so the whole algorithm
runs in linear time. In the case in which $G$ is a positive instance, we obtain a $1$-sided outer-\fb drawing as follows. If~$n<9$, then we directly construct the drawing as in Fig.~\ref{fig:outer-3-recognition}. Otherwise, we identify the case of the proof and then create a drawing according to Fig.~\ref{fig:outer-5-degrees}.
\end{proof}

\section{Conclusions}
\label{sec:conclusions}

In this work, we studied a new drawing model, which introduces the visualization technique of edge-bundling in the framework of beyond-planarity, focusing in particular on the class of fan-planar graphs. 
Our work opens several research directions:
\begin{itemize}
\item Find recognition algorithms for $1$- or $2$-sided (biconnected) outer- or $2$-layer \fb graphs;
\item close the gaps in the density bounds of Table~\ref{table:density};
\item discuss relationships with other classes of nearly-planar graphs;
\item study the \emph{$k$-fan-bundle-planarity}, where each fan-bundle can be crossed at most~$k$ times;
\item consider other models of edge~bundling suitable for theoretical analyses and comparisons, e.g., allowing edges to be bundled together not only at their endpoints.
\end{itemize}

\section*{Acknowledgments} 
\noindent The authors would like to thank the anonymous reviewers for their constructive comments. This work is partially supported by the DFG grants Ka812/17-1 and Schu2458/4-1.

\bibliographystyle{abbrvurl}
\bibliography{abbrv,bundles}

\end{document}